\documentclass{article}
\makeatletter
\newcommand{\labitem}[2]{%
\def\@itemlabel{\text{#1}}
\item
\def\@currentlabel{#1}\label{#2}}
\makeatother
\usepackage{setspace,ctable}
\usepackage{a4wide}
\usepackage{lmodern}
\usepackage{booktabs}

\usepackage{amsmath}
\usepackage{amsthm}
\usepackage{epsfig}
\usepackage{amsfonts}
\newtheorem{theorem}{Theorem}[section]
\usepackage{graphicx}
\usepackage{courier}
\usepackage{colortbl}
\usepackage{multirow}
\RequirePackage[OT1]{fontenc}
\RequirePackage{amsthm}
\usepackage{color}

\usepackage{graphicx} 
\usepackage{amsmath}
\usepackage{amssymb}
\usepackage{float}
\usepackage{enumerate}
\usepackage{booktabs}
\usepackage{marvosym,enumitem}
\usepackage{array, xcolor, lipsum, bibentry}
\usepackage{epstopdf}
\usepackage[font=small,labelfont=bf]{caption}
 \usepackage[round]{natbib}

\definecolor{lightgray}{gray}{0.8}
\newcolumntype{L}{>{\raggedleft}p{0.14\textwidth}}
\newcolumntype{R}{p{0.8\textwidth}}

\newcommand{\vertiii}[1]{{\left\vert\kern-0.25ex\left\vert\kern-0.25ex\left\vert #1
    \right\vert\kern-0.25ex\right\vert\kern-0.25ex\right\vert}}

\newtheorem{lemma}{Lemma}[section]
\newtheorem{assump}{Assumption}[section]

\newcommand{\R}{\mathbb R } 
\newcommand{\s}{\frac{1}{n}\sum_{i=1}^n}
\newcommand{\op}{\mathcal O}
\newcommand{\Op}{\mathcal O_P}
\newtheorem{condition}{Condition}[section]

\theoremstyle{definition}
\newtheorem{remark}{Remark}[section]

\newtheorem{example}{Example}[section]

\title{
{Confidence regions for
high-dimensional 
generalized linear models under sparsity
}
}
  \author{Jana Jankov\'{a} \and Sara van de Geer  
        }

\begin{document}
\maketitle

\setcounter{section}{0}

\begin{abstract}
We study asymptotically normal estimation and confidence regions for low-dimensio\-nal parameters in high-dimensional sparse models. Our approach is based on the $\ell_1$-penalized M-estima\-tor which is used for construction of a bias corrected estimator. We show that the proposed estimator is asymptotically normal, under a sparsity assumption on the high-dimensional parameter, smoothness conditions on the expected loss and an entropy condition. This leads to uniformly valid confidence regions and hypothesis testing for low-dimensional parameters. The present approach is different in that it allows for treatment of loss functions that we not sufficiently differentiable, such as quantile loss, Huber loss or hinge loss functions.
We also provide new results for estimation of the inverse Fisher information matrix, which is necessary for the construction of the proposed estimator.
We formulate our results for general models under high-level conditions, but investigate these conditions in detail for generalized linear models and provide mild sufficient conditions. As particular examples, we investigate  the case of quantile loss and Huber loss in linear regression and demonstrate the performance of the estimators in a simulation study and on real datasets from genome-wide association studies. We further investigate the case of logistic regression and illustrate the performance of the estimator on simulated and real data. 
\vskip 0.5cm
Keywords: Sparsity; High-dimensional; Lasso; Entropy; Generalized linear model; Inverse covariance matrix

\end{abstract}

\section{Introduction}
\noindent 
The need to develop efficient methodology for handling high-dimensional data arises in a variety of  applications including genome-wide studies, image processing and pattern recognition.
Penalized M-estimators have become a popular tool for point estimation in such high-dimensional settings. 
Our goal in this paper however goes beyond point estimation: we aim to construct and study methodology for quantifying the uncertainty of estimation. 
\par 
Suppose that we observe a sample  $X_1,\dots,X_n\in\mathcal X$ of independent observations from an unknown distribution $P$ 
which is known to belong to a class  $\mathcal P = \{P_{\beta}\}$ where 
$\beta$ 
ranges over a subset of $\mathbb R^{p}$. 
We assume that the setting is high-dimensional: the number of unknown parameters  $p$ may be greater than the sample size. 
We denote by $P$ the mean with respect to $P$ (assuming it exists) and by $\mathbb P_n$ the empirical mean given the sample $X_1,\dots,X_n.$
Consider a given loss function $\rho_\beta: \mathcal X  \rightarrow \mathbb R$ and let the true unknown parameter $\beta_0$ be defined as
$$\beta_0 := \text{arg}\min_{\beta\in\mathbb R^p} P\rho_\beta.$$
To estimate $\beta_0$, we consider $\ell_1$-penalized M-estimators defined by
$$\hat\beta := \text{arg}\min_{\beta\in\mathbb R^p} \mathbb P_n \rho_\beta + \lambda \|\beta\|_1.$$
\par
Under restrictions on the number of non-zero coefficients in $\beta_0$ and several technical assumptions (compatibility condition, margin condition), consistency of $\ell_1$-penali\-zed M-estimators can be achieved. For an overview of the theoretical results, we refer the reader to \cite{hds} and the references therein.
The so-called ``oracle inequalities'' show consistency in $\ell_1$-norm and consistency of the excess risk 
 at a near-oracle rate $s\lambda$ and $s\lambda^2$, respectively, where $s$ is the number of non-zero elements of $\beta_0$. 
More precisely, if certain regularity conditions are satisfied (see \cite{hds}), it holds
$$\|\hat\beta-\beta_0\|_1=\mathcal O_P(s\lambda^2), \quad P (\rho_{\hat\beta}-\rho_{\beta_0}) = \mathcal O_P(s\lambda^2),$$ 
where $\mathcal O_P(1)$ means boundedness in probability.
Under somewhat stronger regularity conditions, consistency in $\ell_2$-norm at the near-oracle rate may be obtained.   
\par
The advantage of $\ell_1$-penalized estimators is that due to the geometry of the $\ell_1$-norm, the estimator $\hat\beta$ may have many coefficients set exactly to zero. In this sense, Lasso methods yield ``variable selection''; however, this holds only if certain  restrictive conditions are satisfied. We make the statement about variable selection more precise for the case of linear regression. To this end, we denote the true non-zero set of $\beta_0 $ by $S:=\{i:\beta^0_i\not =0\}$ and its estimated analogue by $\hat S:=\{i:\hat\beta_i\not =0\}$. Then under a ``beta-min condition'', which, loosely speaking, requires the non-zero coefficients in $\beta_0$ to be sufficiently large (i.e. above the noise level), it holds that $P(S\subseteq \hat S)\rightarrow 1.$
If, in addition, the ``irrepresentability condition'' (which a restrictive assumption on the design matrix) is satisfied, it holds 
that $P(\hat S= S)\rightarrow 1.$ 
\par
A disadvantage of Lasso-penalized estimators is that the variable selection properties are only guaranteed under restrictive conditions and the asymptotic distribution of penalized M-estimators is in general not tractable. 
The asymptotic behaviour of $\ell_1$-penalized M-estimators has been studied in several papers, see e.g. 
\cite{knight2000}. They suggest, as  one might expect, that the classical theory on asymptotic normality of M-estimators cannot be immediately regenerated. Nevertheless, penalized estimators might be used as initial estimators to construct estimators that are asymptotically normal and regular under mild conditions, thus moving towards asymptotically efficient estimation. These asymptotically normal estimators may then be used for variable selection in the spirit of the more classical framework of hypothesis testing.
\par
To construct asymptotically normal estimators for (sparse) high-dimen\-sio\-nal models, several different methods have been studied. We mention the idea of bias correction of an initial Lasso estimator, which 
was studied in the papers \cite{zhang,vdgeer13,stanford1} for linear regression and the paper \cite{vdgeer13} considers in addition the generalized linear models. The message of these papers is that for inference about low-dimensional parameters of interest, one needs a good initial estimator of the high-dimensional parameter and an estimator of the score of the nuisance parameter. 
The approach has also been applied in particular examples of nonlinear models, see e.g. \cite{jvdgeer14} and \cite{jvdgeer15} for Gaussian graphical models. An alternative approach, based on Neyman's orthogonalizing conditions, was studied in \cite{vch1}. This yields an asymptotically normal estimator for low-dimensional parameters, by solving the orthogonalizing conditions, using an initial Lasso-regularized estimator. The  paper \cite{vch1} provides high-level conditions under which asymptotic normality can be obtained. The estimator they propose is related to the bias correction idea, although the two approaches are not identical. This approach was also studied for the least absolute deviations estimator in \cite{vch}, where asymptotic normality of the estimator based on Neyman's orthogonalizing conditions was shown.
Further works on inference in high-dimensional settings include \cite{bootstrap1}, \cite{vdgeer14}, \cite{chisq}, \cite{vdgeer12} and other.

\subsection{Contributions}
The paper \cite{vdgeer13} is  closely related to our work. This paper extends the analysis therein to non-diffe\-ren\-tiable loss functions and relaxes certain conditions therein. The paper \cite{vdgeer13} constructs an asymptotically normal estimator for low-dimensional parameter in high-dimensional generalized linear models. This is done by bias correction of an initial Lasso estimator. To calculate the bias correction, nodewise Lasso regressions are used to approximately invert a high-dimensional precision matrix, which corresponds to the inverse Fisher information.
The results of \cite{vdgeer13} assume that the loss function is twice differentiable and the second derivative is Lipschitz. We relax this assumption by considering entropy of the classes of functions instead of Lipschitz properties.  This moves the differentiability and Lipschitz conditions from the loss function onto the \emph{expected} loss function.
We also derive alternative theoretical results for estimating the precision matrix with nodewise Lasso regressions. These results hold for generalized linear models with bounded design under mild conditions, which are alternative to the conditions in the paper \cite{vdgeer13}.
For general high-dimensional models, we provide high-level conditions, which can be checked in particular situations.
The theoretical results are supported by a simulation study and applications to real data from genome studies in linear and logistic regression. 

\subsection{Organization of the paper}
In Section \ref{sec:glm.main} we consider the high-dimensional generalized linear model. 
We describe the de-sparsifying methodology and estimation of the inverse covariance matrix in Section \ref{subsec:method}. 
Main theoretical results for generalized linear models are contained in Section \ref{subsec:mtr}.
Section \ref{sec:nw} contains main theoretical results for nodewise regression for estimation of inverse covariance matrices.
In Section \ref{sec:main.all} we consider general high-dimensional models. 
Examples including the Lasso, least absolute deviations estimator and the Huber estimator are contained in Section \ref{sec:examples}.
Section \ref{sec:conc} summarizes the findings. Sections \ref{sec:sims} and \ref{sec:real} contain simulation studies and applications to real data sets in linear and logistic regression. The proofs are deferred to Appendix \ref{sec:proofs}.

\subsection{Notation}
Let $X_1,\dots,X_n$ be independent random variables with values in some space $\mathcal X$ and let $\mathcal F$ be a class of real valued functions on $\mathcal X.$ 
For a function $f:\mathcal X\rightarrow \mathbb R$ we denote  by $\mathbb P_n f=\sum_{i=1}^n f(X_i)/n$ its empirical measure and by 
$Pf=\sum_{i=1}^n \mathbb Ef(X_i)/n$ its theoretical measure (assuming the integrals exist). Let $\|f\|_n^2 := \mathbb P_n f^2$ denote the empirical norm of $f$
and let $\|f\|=Pf^2$ (assuming it exists) denote the theoretical norm of $f$.
Let $\mathbb G_n f:= \sqrt{n}(\mathbb P_n-P)f.$
By $N(\epsilon, \mathcal F, \|\cdot \|_{n})$ we denote the covering number of the set $\mathcal F$, which is the minimum number of $\|\cdot \|_n$-balls with radius $\epsilon$ needed to cover the set $\mathcal F.$ 
For a vector $x=(x_1,\dots,x_p)\in\mathbb R^p$ we denote its $\ell_r$ norm by $\|x\|_r:= (\sum_{i=1}^p x_i^r)^{1/r}$ for $r\geq 1$. We further let $\|x\|_\infty:=\max_{i=1,\dots,p}|x_i|$ and $\|x\|_0 = |\{i:x_i\not =0\}|.$ 
For a matrix $A$ we denote its $j$-th column by $A_j$  and its $(i,j)$-th element by $A_{ij}$.

\section{High-dimensional generalized linear models}
\label{sec:glm.main}

We are given independent observations $(X_1,Y_1),\dots,(X_n,Y_n)$. $Y_i\in\mathbb R^p$ has the interpretation of the dependent variable and $X_i\in\mathbb R^p$ represents covariates. 
Let $\rho_\beta:\mathbb R^p \times \mathbb R \rightarrow \mathbb R$ for $\beta \in\mathbb R^p$ be a given loss function.
We assume that the loss function depends on the parameter only through the linear combination $x^T\beta$, i.e.
\begin{equation}
\label{loss.g}
\rho_\beta(x,y) := \rho(x^T\beta,y).
\end{equation}
The loss function is not necessarily related to the probability distribution of the instances.  

\begin{example}\textbf{(Generalized linear models)}
\label{ex:glm}
A special case of the above setting is the generalized linear model (\cite{mccullagh}), where 
$\mathbb E(Y_i|X_i) =  g(X_i^T\beta_0).$
Then the probability density function has the form
$p_\beta(y|x)=f(y,g(x^T\beta)),$
for some function  $f$.
If the loss function equals the negative log-likelihood; this corresponds to a maximum likelihood approach.
\hfill 
\end{example}
\noindent
Some examples of loss functions covered in this paper include
\begin{enumerate}[label=(\roman*)]
\item 
quadratic loss $\rho(u,y)=(y-u)^2$,
\item
quantile loss $\rho(u,y) = q|y-u|1_{y-u>0}+(1-q)|y-u|1_{y-u\leq 0},y-u\in \R$, for some $0 \leq q \leq 1$,
\item
Huber loss
$\rho(u,y) = [(y-u)^2 1_{|y-u|\leq K}+K(2|y-u|-K)1_{|y-u|>K}]/(2K)$,
\item
hinge loss
$\rho(u,y) = (1-yu)_{+}$,
\item
mixture models: $\rho(u,y)=\log \left[\pi \frac{1}{\sigma_1}\phi\left(\frac{y-u_1}{\sigma_1}\right) + (1-\pi) \frac{1}{\sigma_2}\phi\left(\frac{y-u_2}{\sigma_2}\right)\right],$
where $u_1=x^T\beta_1,$ $u_2=x^T\beta_2.$
\item
logistic loss:
$\rho(u,y)=-yu + \log(1+e^{u})$.
\end{enumerate}
A given loss function defines an $\ell_1$-penalized M-estimator via
\begin{equation}\label{def.glm}
\hat\beta := \text{arg}\min_{\beta\in\mathbb R^p} \frac{1}{n}\sum_{i=1}^n\rho(X_i^T\beta, Y_i) + \lambda \|\beta\|_1.
\end{equation}
The above optimization problem implies the first order necessary conditions, so called ``\textit{estimating equations}'',  
\begin{equation}
\label{ee}
\s \psi_{\hat\beta}(X_i,Y_i) + \lambda \hat Z =0,
\end{equation}
where $\psi_\beta:\mathcal X\rightarrow \mathbb R^p$ is the sub-differential of $\rho_\beta$ evaluated at $\hat\beta$ and $\hat Z$ is  the sub-differential of the $\ell_1$-norm evaluated at $\hat\beta$:
\[
\hat Z_{i}=
\begin{cases}
\text{sign}(\hat\beta_i) & \text{ if }\hat\beta_i\not = 0,\\
\hat q_i \in [0,1] &\text{ otherwise.} 
\end{cases}
\]
When the loss function is differentiable in $\beta$, equation \eqref{ee} simply applies  with $\psi_\beta:= \dot \rho_\beta$. 
When the loss function is not differentiable in $\beta$, but it is sub-differentiable, one may still replace the derivative by sub-differential. 
Examples of loss functions that are not differentiable (in every point) but the sub-derivative exists at every point include e.g. quantile loss function (used in quantile regression) or hinge loss function (used in support vector machines). 
\par
$\ell_1$-penalized M-estimators have been studied extensively and under certain conditions, they copy the behaviour of an ``oracle'', which knows the true position of zero entries of $\beta_0.$   
The technical conditions for oracle inequalities were briefly outlined in the introduction and we do not treat them in detail in the present paper, as they are well established in literature, see e.g. the book \cite{hds}. We remark that differentiability of the loss function is not necessary for the oracle inequalities.

\subsection{Methodology: De-sparsifying the Lasso}
\label{subsec:method}
We follow the methodology from \cite{vdgeer13}, which implements a bias correction step on the initial Lasso estimator. 
The methodology of bias correction removes the bias associated with the $\ell_1$-penalty and leads to a non-sparse estimator which recovers the desired asymptotic properties that e.g. the maximum likelihood estimator possesses in low-dimensional settings. 

\subsubsection{Establishing an asymptotic pivot}
The estimating equations as in \eqref{ee} read
$$\mathbb P_n \psi_{\hat\beta} + \lambda \hat Z =0.$$
The idea is to find a root $\hat b$ which (approximately) satisfies the estimating equations without the bias term and thus asymptotically behaves as the oracle estimator, which knows the true positions of non-zero entries of $\beta_0$ and applies a maximum likelihood estimator.
This can be done by arguments relying on second order approximations via Taylor expansions. We will proceed in an equivalent way by ``inverting the estimating equations'' with the Hessian matrix of the loss function. To avoid the need to assume differentiability of the loss function, we do the inversion with a matrix that represents the Hessian of the \textit{expected} loss function. To this end, we denote $\Theta:=([P\psi_{\beta}]'_{\beta=\beta_0})^{-1}$.
When the loss function is twice differentiable 
and equals the negative log-likelihood, 
$\Theta$ is the Fisher information matrix.
Multiplying the estimating equations with $\Theta$ and adding $\hat\beta-\beta_0$ to both sides yields
$$\hat\beta-\beta_0 + \Theta\lambda \hat Z =- \Theta\mathbb P_n \psi_{\hat\beta} +\hat\beta-\beta_0.$$
This leads (by rearranging) to 
the following (classical) decomposition (see \cite{vdv}):
\begin{eqnarray}\label{deco}
\hat\beta_j-\beta_j^0 -\Theta_j^T \mathbb P_n \psi_{\hat\beta} 
&=&
- \underbrace{\Theta_j^T\mathbb  P_n \psi_{\beta_0}}_{\text{asymptotic pivot}}
\\[10px]
\nonumber
&&-\underbrace{\Theta_j^T(\mathbb P_n-P)(\psi_{\hat\beta} - \psi_{\beta_0} ) }_{\text{empirical process part}}+
\underbrace{\hat\beta_j-\beta_j^0 -\Theta_j^T P(\psi_{\hat\beta} - \psi_{\beta_0})}_{\text{smoothness part}}
\end{eqnarray}
where $\Theta_j$ is the $j$-th column of $\Theta
$. Contrary to the classical setting as studied in \cite{vdv}, there is an extra term $\Theta_j^T \mathbb P_n \psi_{\hat\beta}$ which corresponds to the bias of the Lasso. Thus the decomposition \eqref{deco} suggests to take a new estimator corrected by the extra term as follows
\begin{equation}
\label{desp}
\tilde b_j := \hat\beta_j -\Theta_j^T \mathbb P_n \psi_{\hat\beta}.
\end{equation}
Clearly the matrix $\Theta$ is typically not known and hence \eqref{desp} is not a proper estimator. In Section \ref{subsec:nodewi} below, we propose an estimator of $\Theta.$
\par
 The decomposition \eqref{deco} is the main tool in our analysis and it illustrates the challenges underlying this problem. 
Provided that the empirical process part and the smoothness part are small enough, of small order $1/\sqrt{n},$ then asymptotic normality of $\sqrt{n}(\tilde b_j-\beta_j^0)$ can be established by classical arguments, under certain conditions on the pivot term. 
\par
The
\textit{empirical process part} in \eqref{deco}  is related to the complexity of the considered class of functions, which are indexed by a \textit{sparse} parameter $\beta.$ 
Our aim in this part is to show entropy bounds for the class
$$\mathcal G:= \{x\mapsto \Theta_j^T (\psi_{\beta}-\psi_{\beta_0}): \beta\in \mathcal B\},$$
where $\mathcal B$ is some sparse subset of $\mathbb R^p$ that will be specified later.
The sparseness of the index set $\mathcal B$ is crucial, in view of our results relying on entropy numbers. 
For an overview of results on entropy numbers, we refer the reader to e.g. \cite{weak}.
\par
The
\textit{smoothness part} in \eqref{deco} is related to the smoothness of the (derivative of the) expected loss function. The smoothness part poses a problem in high-dimensional settings since Taylor expansions have to be carried out with caution because norms in $\mathbb R^p$ are not equivalent asymptotically when $p\rightarrow \infty$. 

\subsubsection{Nodewise Lasso regression}
\label{subsec:nodewi}
The next challenge is that the high-dimensional vector $\Theta_j\in\mathbb R^p$ is unknown and has to be replaced by a well-behaved estimator.
In view of the decomposition \eqref{deco}, consider now the following second decomposition, given that $\hat\Theta_j$ is an estimate of $\Theta_j$,
\begin{eqnarray*}
\hat\beta-\beta_0 -\hat\Theta_j^T \mathbb P_n \psi_{\hat\beta}
=
\hat\beta-\beta_0 -\Theta_j^T \mathbb P_n \psi_{\hat\beta}-(\hat\Theta_j-\Theta_j)^T \mathbb P_n \psi_{\hat\beta}.
\end{eqnarray*}
This implies that in order for the remainder term to be negligible, $\hat\Theta_j$ must satisfy the condition
 $$\|(\hat\Theta_j-\Theta_j)^T \mathbb P_n \psi_{\hat\beta}\|_\infty =o_P (1/\sqrt{n}).$$
By H\"older's inequality, we have the bound $\|(\hat\Theta_j-\Theta_j)^T \mathbb P_n \psi_{\hat\beta}\|_\infty
\leq \|\hat\Theta_j-\Theta_j\|_1 \|\mathbb P_n \psi_{\hat\beta}\|_\infty.$
Hence an appropriately fast rate for $\hat\Theta_j-\Theta_j$ in $\ell_1$-norm and approximately satisfied estimating equations are sufficient.
To this end, note that the estimating equations \eqref{ee} imply $\|\mathbb P_n \psi_{\hat\beta}\|_\infty =\Op (\lambda).$
\par
In view of the above arguments, our goal is to construct an $\ell_1$-oracle estimator of the inverse of $\Theta.$ 
Estimation of $\Theta$ was well explored in literature for the case of quadratic loss; this is the same problem as estimation of the edge weights in undirected graphical models (see e.g. \cite{glasso}, \cite{buhlmann}, \cite{yuan}).
The challenge arises because, in general, $\Theta$ depends on $\beta_0$ (an exception is the quadratic loss).
We employ the fact that for generalized linear models, there is a special structure in the matrix $\Sigma:=\Theta^{-1} \equiv(P\psi_\beta)'_{\beta=\beta_0}$. 
We can then use as an estimator of $\Sigma$ the empirical version
$$\hat\Sigma_{\hat\beta} = \s  \ddot \rho(y_i,x_i^T\hat\beta)x_ix_i^T.$$
Denoting the weights by $\hat { w}_i:=\ddot \rho(y_i,x_i^T\hat\beta)$ and $\hat W\equiv W_{\hat\beta}:=\text{diag}(\hat w_i)$; then we may rewrite $\hat\Sigma_{\hat\beta} =X^T\hat W^2 X/n.$
This matrix is not invertible because of the high-dimensional setting, but we can approximately invert it using nodewise regression of each column of $\hat WX$ on all the other columns.
To construct in this way a nodewise regression estimator of $\Theta$, we define 
\begin{equation}\label{parcor}
\hat\gamma_{\hat\beta,j} := \text{arg}\min_{\gamma_j\in\mathbb R^{p-1}}\|\hat W_{}(X_j -X_{-j}\gamma_j)\|_2^2/n + 2\lambda \|\gamma_j\|_1,
\end{equation}
$$\hat\tau_{\hat\beta,j}^2 := \|\hat W_{}(X_j -X_{-j}\hat\gamma_j)\|_2^2/n + \lambda \|\hat\gamma_j\|_1,$$
and
\begin{equation}\label{def.nw}
\hat\Theta_{{}j}:= (-\hat\gamma_{\hat\beta,j,1},\dots, -\hat\gamma_{\hat\beta,j,j-1},1,-\hat\gamma_{\hat\beta,j,j+1},\dots,
-\hat\gamma_{\hat\beta,j,p})/\hat\tau_{\hat\beta,j}^2.
\end{equation}
The estimator $\hat\Theta_{j}$ was studied in the paper \cite{vdgeer13} for generalized linear models. In the next sections, we provide alternative conditions under which the methodology yields good estimators.
For generalized linear models with non-differentiable loss functions, other methods have to be used (this is discussed in Section \ref{sec:nw.nondif}). 

\begin{remark}\label{sqrtnode}
Instead of the Lasso, one could use the square-root Lasso (\cite{sqrtlasso}) to estimate the partial correlations in \eqref{parcor}:
\begin{equation}\label{parcor2}
\hat\gamma_{\hat\beta,j} := \text{arg}\min_{\gamma_j\in\mathbb R^{p-1}}\|\hat W_{}(X_j -X_{-j}\gamma_j)\|_2/n + 2\lambda \|\gamma_j\|_1,
\end{equation}
 The advantage is that the square-root Lasso automatically estimates the noise variance as well and thus uses a universal choice of the tuning parameter which is particularly useful from a practical point of view.  To avoid digressions, we do not elaborate on the theoertical results for this alternative method in the present paper.
\end{remark}

\par
Finally, using the nodewise regression estimator $\hat\Theta=(\hat\Theta_1,\dots,\hat\Theta_p),$ we may define the new corrected estimator
$$\hat b:= \hat\beta - \hat\Theta^T \mathbb P_n \rho_{\hat\beta}.$$
This estimator will be referred to as the \textit{de-sparsified Lasso}, in line with \cite{vdgeer14}. In some literature, it is called the ``de-biased Lasso''.

\subsection{Main theoretical results} 
\label{subsec:mtr}

\subsubsection{Model  assumptions}

\begin{enumerate}
\labitem{(A1)}{itm:bounded} (Observations)
Assume that $(X_i,Y_i)$ are independent for $i=1,\dots,n$ and identically distributed for each fixed $n$.
Suppose that $\|X_i\|_\infty \leq K_X,i=1,\dots,n$, $\mathbb E|\Theta_j^T X_i|^4 =\mathcal O(1)$, $1/\Lambda_{\min}(\Sigma) = \mathcal O(1)$. Define $\beta_0$ by $P\psi_{\beta_0}=0$ and assume that $1/ (\Theta_j^TP\psi_{\beta_0}\psi_{\beta_0}^T\Theta_j)= \mathcal O(1)$.
\labitem{(A2)}{itm:initial.glm} (Initial estimates)
Suppose that 
$\|\hat\beta-\beta_0\|_1 \leq Cs\lambda,\mathbb E\|X(\hat\beta-\beta_0)\|^2_2/n\leq Cs\lambda^2$ 
and $\|\hat\beta\|_0\leq Cs$ with high probability for some $C>0.$
\end{enumerate}
We remark that condition \ref{itm:bounded} assumes that the observations are identically distributed for every fixed $n$. This is not important for the analysis, we only assume this to keep the presentation cleaner. 
The moment condition $\mathbb E|\Theta_j^T X_i|^4 =\mathcal O(1)$ holds e.g. for  a sub-Gaussian random vector $X_i$ (see \cite{hds}) when $\Lambda_{\max}(\Theta)=\mathcal O(1)$. Clearly, it is a much weaker condition than requiring sub-Gaussianity of $X_i$.
Furthermore, as already noted, the rates of convergence $\|\hat\beta-\beta_0\|_1 \leq Cs\lambda,$ $\|X(\hat\beta-\beta_0)\|^2_2/n\leq Cs\lambda^2$  from condition \ref{itm:initial.glm} were derived under mild conditions in the book by \cite{hds}.
As for the condition $\|\hat\beta\|_0\leq Cs$, we show in Theorem \ref{sparsity.glm} below that this is satisfied under mild conditions.
\par
We introduce some further notation.  
If $u\mapsto \rho (y,u)$ is differentiable at $u$ then we 
denote $w(y,u) := \frac{\partial \rho(y,u)}{\partial u}.$  We call $w$ the \textit{weight function} or \textit{weight}.
For illustration, we give a few examples of weight functions:
\begin{enumerate}[label=(\roman*)]
\item  quadratic loss: $w_{\text{quadratic}}(y,u)=2(y-u),$
\item
 absolute loss: $w_{\text{absolute}}(y,u)=\text{sign}({y-u}),$
for all $u\not =y.$
\item
 Huber loss:
\vskip -0.5cm
\[
w_{\text{Huber}}(y,u)=
\begin{cases}
(y-u)/K & \text{if }|y-u|\leq K,\\
\text{sign}(y-u)& \text{ otherwise.}
\end{cases}
\]
\end{enumerate}

\subsubsection{Main results for differentiable loss functions}
In this section we consider differentiable loss functions as summarized in the following conditions.
\begin{enumerate}
\labitem{(B1)}{itm:sparsity.glm}(Sparsity)
 Let $\|\Theta_j\|_0\leq s$, where it holds
  $s^3(\log p)^2 (\log n)^2/n\rightarrow 0.$
\labitem{(B2)}{itm:w}\label{w}
(First order differentiability)
Assume that  $w(y,u)$ is Lipschitz in $u$ with a constant $L=\mathcal O(1)$ for all $\beta$ such that $\|\beta-\beta_0\|_1 \leq \delta$
for some $\delta>0$. 
\labitem{(B3)}{itm:wdot}
(Second order differentiability)
Suppose that $u\mapsto w(y,u)$ is differentiable and its derivative $\dot w$ is Lipschitz in $u$ with $L=\mathcal O(1)$
for all $\beta$ such that $\|\beta-\beta_0\|_1 \leq \delta$
for some $\delta>0$.
\end{enumerate}
Note that assumption \ref{itm:w} rules out the absolute loss and assumption \ref{itm:wdot} rules out the Huber loss. Relaxations of \ref{itm:w} will be treated in Section \ref{subsec:glm.nd} below.
We remark that the analysis in \cite{vdgeer13} requires both \ref{itm:w} and \ref{itm:wdot}, i.e. that the loss function is twice differentiable and the second derivative is Lipschitz. The conditions \ref{itm:w} and \ref{itm:wdot} are stated separately only in view of  Theorem \ref{glm.main} below. We further need conditions \ref{eig}, \ref{A3}, which are needed for the estimation of the score for the nuisance parameter as given in Section \ref{sec:nw}.
The following lemma, which underlies the theoretical result of Theorem \ref{glm.main} below, gives an bound on entropy of a certain class of functions.

\noindent
\begin{lemma}
\label{glm.entropy.sparse}
Suppose that condition \ref{itm:w} is satisfied.
For some constant $C_1>0$ let
$$\mathcal F:= \{x\mapsto \Theta_j^T x (w_{\beta}-w_{\beta_0}): \|\beta-\beta_0\|_{1} \leq \delta, \|\beta\|_0 \leq C_1 s \},$$
Then for all $\epsilon>0$ 
$$N(\mathcal F, \|\cdot\|_{n},\epsilon) \leq C_2 \left( \frac{ps}{\epsilon} \right)^s,$$
and
$$\log N(\mathcal F, \|\cdot\|_{n},\epsilon) \leq C_3 s\log p + C_3s\log \left(\frac{1}{\epsilon}\right),$$
for some constants $C_2,C_3>0.$

\end{lemma}

\begin{theorem}\label{glm.main}
Let $\rho_\beta$ be a given loss function and assume that Conditions \ref{itm:bounded}, \ref{itm:initial.glm}, \ref{itm:sparsity.glm} and \ref{itm:w} are satisfied with $\psi_\beta=\dot \rho_\beta$. 
Let $\hat\beta$ be defined by \eqref{def.glm}.
 Define 
 $$\tilde b_j : = \hat\beta_j - \Theta_j^T\mathbb P_n \dot \rho_{\hat\beta}.$$
Then
$$\sqrt{n}(\tilde b_j - \beta^0_j)/
\sqrt{\Theta_j^T P\dot\rho_{\beta_0}\dot\rho_{\beta_0}^T \Theta_j} \rightsquigarrow \mathcal N(0,1).$$
Let $\hat\Theta_j$ be defined in \eqref{def.nw}. In addition, if conditions \ref{itm:wdot}, \ref{eig}, \ref{A3} hold, then
$$\|\hat\Theta_j-\Theta_j\|_1 = \Op (s^{3/2}\sqrt{\log p/n}),$$
and the estimator
 $$\hat b_j : = \hat\beta_j - \hat\Theta_j^T\mathbb P_n \dot \rho_{\hat\beta}$$
satisfies
$$\sqrt{n}(\hat b_j - \beta^0_j)/
\sqrt{\hat\Theta_j^T \mathbb P_n\dot\rho_{\hat\beta}\dot\rho_{\hat\beta}^T \hat\Theta_j} \rightsquigarrow \mathcal N(0,1).$$

\end{theorem}

Compared to Theorem 3.1 in \cite{vdgeer13}, Theorem \ref{glm.main} does not relax the differentiability conditions. 
This differentiability result is however only needed for estimation of the score for the nuisance parameter $\Theta_j$ as can be seen from the first part of Theorem \ref{glm.main}. 
Thus the differentiability of the loss function can be relaxed to first derivative being Lipschitz, provided that a good   estimate of $\Theta_j$ (in $\ell_1$-norm) is available. 
\par
The second part of the theorem illustrates conditions for nodewise Lasso which are alternative to the conditions in \cite{vdgeer13}. 
We do not need the condition $\|X \Theta_j\|_\infty\leq K$ (where $K>0$ is a constant).
The price we pay for this relaxation is a stronger sparsity assumption $s^{3/2}=o(\sqrt{n}/\log p)$.  

\par
An application of the  result of Theorem \ref{glm.main} can be considered for the case of Huber loss. Huber loss is once differentiable and the derivative is Lipschitz continuous. The second derivative exists everywhere except $|u-y|=K$, however, it is not Lipschitz. Hence the results of \cite{vdgeer13}
do not apply to Huber loss. The first part of the Theorem \ref{glm.main} does apply, and one then needs to estimate $\Theta$.  This is treated in Section \ref{subsec:ex.huber}.

\subsubsection{Main results for non-differentiable loss functions}
\label{subsec:glm.nd}

\indent 
Theorem \ref{glm.main} however still does not cover an important example such as the quantile regression due to assumption \ref{itm:w}.
The absolute loss is differentiable everywhere except $u=y$, however, the first derivative is not Lipschitz.
In the conditions below, we do not require that $w$ is Lipschitz, but we require that its expectation is Lipschitz, which is a much weaker assumption.
We formulate a relaxation of Theorem \ref{glm.main} to non-differentiable functions below. 
This requires an entropy condition  on the class of functions which are related to the empirical process part of the problem.

\begin{enumerate}
\labitem{(C1)}{itm:ew}\label{ew}
Suppose that the function $\psi_\beta(y,x)$ has the form $\psi_\beta(y,x) = w(y,x^T\beta)x$ for some function $w$.
Assume that the function $u\mapsto G(u):= \int w(u,y)dP_{Y|X}$ is differentiable and $u\mapsto G'(u)$ is Lipschitz.
Suppose that for some $\delta>0,$ the function  $w_\beta$ is bounded from above and stays away from zero uniformly in $n$ for all $\beta$ that satisfy $\|\beta-\beta_0\|_1\leq \delta.$
\labitem{(C2)}{itm:other}
Suppose that 
 $\mathbb E (w_\beta-w_{\beta_0})^2=\mathcal O\left({s^{-2} (\log p)^{-2} (\log n)^{-4}}\right)$ for all $\|\beta-\beta_0\|_1\leq \delta$ for some $\delta>0$.
\end{enumerate}

\begin{theorem}\label{glm.main2}
Assume conditions \ref{itm:bounded}, \ref{itm:initial.glm}, \ref{itm:ew} and \ref{itm:other} with some function $\psi_\beta$.
Suppose that the function $\psi_\beta$ and $\hat\beta\in\mathbb R^p$ satisfy
$$\|\mathbb P_n \psi_{\hat\beta}\|_\infty = \mathcal O_P(\lambda).$$
Suppose that $\|\Theta_j\|_0\leq s$ and $\hat\Theta_j$ satisfies $\|\hat\Theta_j-\Theta_j\|_1 =\mathcal O_P(s\lambda)$. 
Consider the class of functions 
$\mathcal F:=\{\Theta_j^T (\psi_\beta - \psi_{\beta_0}): \|\beta-\beta_0\|_1 \leq s\lambda,
\mathbb E\|X(\hat\beta-\beta_0)\|_2^2/n \leq s\lambda^2, \|\beta\|_0 \leq s\}$.
Suppose that 
\begin{equation}
\label{ent.glm}
\log N(\mathcal F, \|\cdot\|_n, \epsilon) \leq s\log \left(\frac{p\|F\|_n}{\epsilon}\right),
\end{equation}
where $F(x)=\sup_{f\in\mathcal F} |f(x)|$ is the envelope function of $\mathcal F.$
Then
$$\sqrt{n}(\hat\beta_j + \hat\Theta_j^T\mathbb P_n \psi_{\hat\beta} - \beta^0_j)/
\sqrt{\Theta_j^T P\psi_{\beta_0}\psi_{\beta_0}^T \Theta_j} \rightsquigarrow \mathcal N(0,1).$$

\end{theorem}

The theorem replaces the differentiability assumption by an entropy condition (see \cite{weak} for similar arguments).
The theorem assumes that we can estimate the score for the nuisance parameter; this is discussed in Section \ref{sec:nw}.

\subsubsection{Sparsity of the Lasso}
\label{subsec:spa}
The following lemma shows that under mild conditions, the $\ell_1$-penalized M-estimator $\hat\beta$ has sparsity of the same order as 
$\beta_0$ with high probability. It is worth out point out that we  do not require differentiability of the loss function. 
\begin{lemma}\label{sparsity.glm}
Suppose that  condition \ref{itm:bounded} holds.
Let $\hat\beta$ be defined by $\mathbb P_n \psi_{\hat\beta} + \lambda\hat Z=0$, where $\hat Z$ is the sub-differential of the $\ell_1$-norm evaluated at $\hat\beta$ and  where the function $\beta\mapsto \psi_\beta$ satisfies condition
\ref{itm:ew} with a function $G$ such that $|G'|\leq K$ for some constant $K>0$. Assume that $\|X(\hat\beta-\beta_0)\|_n^2 = \mathcal O_P(s\lambda^2)$ and $\|\hat\beta-\beta_0\|_1=\mathcal O_P(s\lambda).$ Further assume that \eqref{ent.glm} is satisfied. 
Then 
$$\|\hat\beta \|_0 =\mathcal O_P(s).$$
\end{lemma}

\subsubsection{Estimation of asymptotic variance}
\label{subsec:var}
To construct confidence intervals, one needs to estimate the asymptotic variance of the de-sparsified estimator.
The following lemma shows that we can use $\hat\Theta_j^T P_n\psi_{\hat\beta}\psi_{\hat\beta}^T \hat\Theta_j$  as an estimator of the asymptotic variance of the de-sparsified estimator, where $\hat\Theta$ is the nodewise regression estimator. 
\begin{lemma}\label{var}
Assume conditions \ref{itm:bounded}, \ref{itm:initial.glm}, \ref{itm:sparsity.glm}, \ref{itm:w} and suppose that 
$\|\hat\Theta_j - \Theta_j\|_1=\Op (s\lambda^2).$ Then 
$$|\hat\Theta_j^T P_n\psi_{\hat\beta}\psi_{\hat\beta}^T \hat\Theta_j - \Theta_j^T P\psi_{\beta_0}\psi_{\beta_0}^T \Theta_j| = o_P(1).$$
\end{lemma}

\section{Nodewise regression for estimation of precision matrices
}
\label{sec:nw}
 Our goal in this section is to provide estimators for 
$\Theta:=\Sigma^{-1}$, 
where 
$$\Sigma:=(\mathbb E\psi_\beta(x,y))'_{\beta=\beta_0}.$$
If the parameter of interest is $\beta_j^0$ for some $j$, we only need to estimate $\Theta_j.$
In the next sections, we suggest procedures for estimation of $\Theta$ in generalized linear models.
We first consider first the case when the loss function is differentiable and we also discuss the case when it is not differentiable.

\subsection{Generalized linear models with differentiable loss functions}
If the loss function is twice differentiable, then 
$$\Sigma =\mathbb E \ddot \rho(y,x^T\beta_0)xx^T= \mathbb E \dot w (y,x^T\beta_0) xx^T.$$
Hence we can approximate 
$\Sigma$ by the empirical version 
$$\hat\Sigma_{\hat\beta}:=\s \dot w_{\hat\beta}(y_i,x_i) x_ix_i^T.$$
This matrix is not invertible in high-dimensional settings, but we can use e.g. nodewise regression to approximately invert it, as outlined in Section \ref{sec:glm.main}. The main difficulty here is that the estimator depends on the estimator $\hat\beta$. 
\\
We formulate the results for general weights satisfying the conditions below. Let the weight matrix $ W_{\hat\beta}$ be given by 
$$ W_{\hat\beta} := \text{diag}(v(y_i,x_i^T\hat\beta))_{i=1,\dots,n},$$
for some weight function $v$. 
Let $\Sigma :=\mathbb E v (y,x^T\beta) xx^T$ and $\Theta=\Sigma^{-1}.$
Below we provide theoretical guarantees for the estimator $\hat\Theta_{j}$ defined in equation \eqref{def.nw}, with the weight matrix $W_{\hat\beta}$.
We make the following assumptions.
\begin{enumerate}[label=(E\arabic*),start=1]
\item
\label{eig}
The matrix inverse $\Theta:=\Sigma^{-1}$ exists and $\max_{j=1,\dots,p} \|\Theta_j\|_0 \leq s.$
Moreover,
$1/\Lambda_{\min}(\Sigma)=\mathcal O(1)$, $\|\Sigma\|_\infty = \mathcal O(1)$.  
\item
\label{A3}
There exists some $\delta>0$ such that for all $\|\beta-\beta_0\|_1 \leq \delta$ it 
holds that $1/v_\beta=\mathcal O(1)$ and $v_\beta =\mathcal O(1).$
Furthermore, $v_\beta$ is Lipschitz with a universal constant.
\item
\label{sp}
It holds that $s^{3/2}\sqrt{\log p /n}=o(1).$
\end{enumerate}

\begin{theorem}\label{my}
Suppose that Conditions \ref{itm:bounded}, \ref{itm:initial.glm} and \ref{eig}-\ref{sp} are satisfied. 
Then using $\lambda_j\asymp  \sqrt{\log p /n}$ for the nodewise Lasso estimator $\hat\Theta_{j}$ defined in \eqref{def.nw} it holds that
$$\|\hat\Theta_{j} - \Theta_{j}\|_1 =\mathcal O_P( s^{3/2}  \sqrt{\log p/n}). $$
\end{theorem}

\noindent
Theorem \ref{my} relaxes the condition $\|X\Theta_j\|_\infty = \mathcal O(1)$ from \cite{vdgeer13}.
Furthermore, we remark that it is a more general result than in \cite{vdgeer13} in that the latter only considers that $\Sigma=\mathbb E\dot wxx^T$, but we allow for any $\Sigma = \mathbb E v xx^T$ for arbitrary weights $v$ satisfying the conditions. From the point of view of the proof, this makes no actual difference, however, the application of the result is then somewhat broader, 
as will be illustrated in Section \ref{subsec:o2} below.

\subsection{Generalized linear models with non-differentiable loss functions}
\label{sec:nw.nondif}
If the loss function is not differentiable, the above strategy clearly cannot be used to estimate $\Theta$.
We discuss some alternative options that could be used.
\subsubsection{Special cases}
In some settings, one can make use of the structure in $(P\psi_\beta)'_{\beta_0}.$ In particular, we have
$$(P\psi_\beta)' = \mathbb E_X (\mathbb E_Y w(u,Y_i))'_{u=x_i^T\beta} x_ix_i^T.$$
The above can be approximated by 
$$\hat \Sigma _{\hat\beta}:=\s (\mathbb E_Y w(u,Y_i))'_{u=x_i^T\hat\beta} x_ix_i^T.$$
In some situations, it is possible to calculate 
$(\mathbb E_{Y}w(u,Y_i))'_{u=x^T\beta_0}$ provided that we assume the distribution of $Y$, and then we can plug in an estimate $\hat\beta$ of $\beta_0$. 
Then we can  estimate $\Sigma$ by $ \s (\mathbb E_{Y}w(u,Y_i))'_{u=x^T\hat\beta} x_i x_i^T$.
We can then use nodewise regression with weights, and under some conditions on the weights, the nodewise regression yields good estimators of $\Theta$, see Theorem \ref{my} in Section \ref{sec:nw}.
For instance, 
for absolute loss in the linear model, we have (if $\epsilon$ and $X$ are independent)
$$(\mathbb E w_{\text{absolute}}(y,x^T\beta))'_{\beta=\beta_0} = 2f_{\epsilon}(0),$$ 
and for Huber loss 
$$(\mathbb E w_{\text{Huber}}(y,x^T\beta))'_{\beta=\beta_0} = F_{\epsilon}(K)-F_\epsilon(-K).$$

\subsubsection{Maximum likelihood}
\label{subsec:o2}
If the loss function equals the negative log-likelihood, we can consider the following approach. Denote the score function by
 $s_\beta:=\frac{\partial \rho}{\partial \beta}.$
Then (for differentiable loss) the following identity holds
$$\mathbb E s_\beta s_\beta^T = -\mathbb E\dot s_\beta=-\mathbb E\ddot \rho_{\beta}.$$
This implies that $\hat\Sigma$ has the form of a Gram matrix with inner products given by score functions corresponding to individual parameters. 
Hence we could use as an alternative estimator of $\Sigma$
$$\frac{1}{n}\sum_{i=1}^n s_{\hat \beta} (X_i) s_{\hat \beta}(X_i)^T=\s w^2(y_i,x_i^T\hat\beta) x_ix_i^T.$$
This again has the form $X^T \hat W X$, so we can do nodewise regression,
if conditions \ref{eig}-\ref{sp} are satisfied with weight function $v := w^2.$
For instance, for absolute loss, $w^2(y,u)=1.$ For Huber loss, $w^2(y,u) = 1$ if $|u|\leq K$ and $w^2(y,u)=u^2$ otherwise.
This function is Lipschitz, and hence Theorem \ref{my} can be applied.

\section{General high-dimensional models}
\label{sec:main.all}

\noindent
In this section we provide results for general models under high-level conditions. 
These are useful for insight into the underlying machinery and its limits.
Furthermore, they can be used to obtain results 
for more general models than the generalized linear models.
Assume we have independent data $X_1,\dots,X_n\in\mathcal X.$ 
We make the following assumptions.

\begin{enumerate}
\labitem{(D1)}{itm:eig0}\label{eig0}
The function $\beta\mapsto P\psi_{\beta}$ is differentiable with a matrix of first derivatives $\Sigma:=(P\psi_\beta)'|_{\beta=\beta_0}$, which satisfies the eigenvalue condition
$$ \Lambda_{\min}(\Sigma) \geq c,$$
for some $c>0.$
Denote $\Theta := \Sigma^{-1}$. 
\labitem{(D2)}{itm:taylor}\label{taylor}
Suppose that the following expansion holds
$$\|(P\psi_\beta - P\psi_{\beta_0}) - \Sigma(\beta-\beta_0)\|_\infty =\mathcal O(d^2(\beta,\beta_0)),$$
where $d$ is some metric.
\labitem{(D3)}{itm:initial}
\label{initial} 
Suppose that $\hat \Theta_j$ satisfies
$\|\hat\Theta_j-\Theta_j\|_1 = \mathcal O_P( s\lambda).$ 
\labitem{(D4)}{itm:empp}
It holds that
\begin{equation*}\label{ep}
\mathbb G_n\Theta_j^T(\psi_{\hat\beta} - \psi_{\beta_0}) =o_P(1).
\end{equation*}
\labitem{(D5)}{itm:norm}
It holds that
$$\sqrt{n}\Theta_j^T \mathbb P_n \psi_{\beta_0}
/\sqrt{\Theta_j^T P\psi_{\beta_0}\psi_{\beta_0}^T \Theta_j} \rightsquigarrow \mathcal N(0,1).$$
\end{enumerate}
\noindent

\begin{theorem}\label{mai}
Let $\psi_\beta:\mathcal X\rightarrow \mathbb R^p$ satisfy $\|\mathbb P_n \psi_{\hat\beta}\|_\infty=\mathcal O_P(\lambda)$ in a given $\hat\beta\in\mathbb R^p$ and define $\beta_0$ by $P\psi_{\beta_0}=0.$
Assume that 
 conditions  \ref{eig0}, \ref{taylor}, \ref{itm:initial}, \ref{itm:empp}, \ref{itm:norm} are satisfied with some function $d(\beta,\beta_0)$.
Denote $s:=\|\Theta_j\|_0$ and assume the sparsity condition 
\begin{equation}
\label{spa.prop}
s=o_P(\max\{ n^2d^2(\hat\beta,\beta_0) , 1/(\sqrt{n}\lambda^2)  \}).
\end{equation}
Then
$$\sqrt{n}(\hat\beta_j + \hat\Theta_j^T\mathbb P_n \psi_{\hat\beta} - \beta^0_j)/
\sqrt{\Theta_j^T P\psi_{\beta_0}\psi_{\beta_0}^T \Theta_j} \rightsquigarrow \mathcal N(0,1).$$
\end{theorem}

The conditions on the initial estimator $\hat\beta$ are in the assumption $s=o_P( n^2d(\hat\beta,\beta_0)^2)$ and in the empirical process condition \ref{itm:empp}.
The sparsity condition \eqref{spa.prop} has two parts: the first part $s=o_P(n^2d(\hat\beta,\beta_0)^2 )$ ensures that there is enough continuity in the problem and the second part 
$s=o_P(1/(\sqrt{n}\lambda^2))$ ensures that the estimator of the score for the nuisance parameter is good enough.
\noindent

Theorem \ref{mai} assumes asymptotic equicontinuity \ref{ep}. The following Theorem shows sufficient conditions for the asymptotic equicontinuity to hold. 
\begin{theorem}\label{main2}
Suppose that for the class of functions $\mathcal F$
it holds that
$$\log N(\mathcal F, \|\cdot\|_n, \epsilon) \leq s\log \left(\frac{p\|F\|_n}{\epsilon}\right),$$
where $F=\sup_{f\in\mathcal F} |f|$ is the envelope function of $\mathcal F.$
Let $R:=\sup_{f\in\mathcal F} \|f\|$ and suppose that 
\begin{equation}
\begin{array}{ccc}\label{condi}
&R\log n\sqrt{s\log p/n} =o(1/\sqrt{n}),& \\[10px]
& (\mathbb E\sup_{f\in\mathcal F} f^4(X_1))^{1/4}(\log n\sqrt{s\log p/n})^{3/2}=o(1/\sqrt{n}).&
\end{array}
\end{equation}
Then
$$
\sup_{f\in\mathcal F}\mathbb G_n f =o_P(1).
$$
\end{theorem}
\noindent
We aim to apply Theorem \ref{main2} with the class of functions
$$\mathcal F:=\{\Theta_j^T (\psi_\beta - \psi_{\beta_0}): \beta\in\mathcal B
\},$$ 
for some set $\mathcal B\subset \mathbb R^p$ which can be specified depending on the problem at hand (for generalized models, see e.g. Theorem \ref{glm.main}), but we must ensure
that $\hat\beta\in\mathcal B$ and at the same time, the set $\mathcal B$ must be in some sense sparse.
\\
Conditions \eqref{condi} are discussed in Section \ref{subsec:disc} below. They mean that the higher order remainders from the Dudley's integral are small enough. The idea is it should be possible to get a rate for these remainders over the set $\mathcal F,$ since the set is shrinking with $n$. They can be shown for generalized linear models under some sufficient conditions.
Combining Theorem \ref{mai} with Theorem \ref{main2} gives explicit sufficient conditions under which asymptotic normality can be achieved.

\noindent

\subsection{Discussion of the conditions}
\label{subsec:disc}

\textit{Condition \ref{itm:eig0}.}
Condition \ref{itm:eig0}  avoids the need to assume differentiability  of $\rho_\beta$ or $\psi_\beta$ directly. Instead we assume differentiability of the expected loss.
Note that in some situations, the matrix $\Theta$ may not exist, for instance for the linear model with fixed design.
To be able to describe the asymptotics, we may then assume that there exists a non-singular matrix $\Theta$ (with eigenvalues  bounded from above and away from zero) such that
\begin{equation}\label{fixed}
\|\Theta_j^T (P\psi_\beta)'|_{\beta=\hat\beta} - e_j\|_\infty=o(1/\sqrt{n}).
\end{equation}
Consequently, replacing assumption \ref{itm:eig0} in Theorem \ref{mai} with \eqref{fixed}, the result of Theorem \ref{mai} applies.

\textit{Condition \ref{itm:taylor}.}
In condition \ref{itm:taylor}, $d(\beta,\beta_0)$ represents a metric suitable for the problem at hand. For generalized linear models, one may choose $d^2(\beta,\beta_0)=\mathbb E| x^T(\beta-\beta_0)|^2 = (\beta-\beta_0)^T\Sigma (\beta-\beta_0)$. For general models, if the function $\beta\mapsto \frac{\partial ^2 (P\psi_\beta)_{j}}{\partial \beta_k\partial \beta_i}$ is bounded for all $i,j,k=1,\dots,p$, one may choose $d^2(\beta,\beta_0) = \|\beta-\beta_0\|_1^2$ (see Lemma \ref{lipschitz}). 
\par
\textit{Sparsity conditions \eqref{spa.prop}.}
To have sufficient continuity in the model as described in condition \ref{taylor}, some sparsity assumptions must be made.
Naturally, considering more general models costs more. 
In general, we require $d^2(\hat\beta,\beta_0)=o\left(\frac{1}{\sqrt{ns}}\right).$ For generalized linear models this conditions amounts to
$s^3(\log p)^2 /n\rightarrow 0,$ even for non-differentiable loss, provided that the expected loss is differentiable.
\par
\textit{Condition \eqref{condi}.}
Theorem \ref{main2} suggests that we need some rate on $R=\sup\|f\|$, in particular it must be shown that $R=o\left(\frac{1}{\sqrt{s\log p}\log n}\right)$. This can indeed be shown for e.g. generalized linear models, under sufficient sparsity conditions (see Section \ref{sec:glm.main}).
\\
Condition $(P \sup_{f\in\mathcal F} f^4)^{}  (\log n)^{6} (s\log p)^{3} / n^{} =o(1)$  is satisfied e.g. for Lipschitz $\psi_\beta$. An envelope function for 
the class $\mathcal F:=\{\Theta_j^T (\psi_\beta-\psi_{\beta_0}): \|\beta -\beta_0\|_1\leq s\lambda \}$ is then obtained using the following upper bounds
\begin{eqnarray*}
|\Theta_j^T (\psi_\beta-\psi_{\beta_0})| &\leq & \|\Theta_j\|_1 \|\psi_\beta-\psi_{\beta_0}\|_\infty
\\ 
&\leq &
\|\Theta_j\|_1 L\|\beta-\beta_0\|_1
\\
&\leq & \Lambda_{\max}(\Theta) L  \sqrt{s}\|\beta-\beta_0\|_1 \leq \Lambda_{\max}(\Theta) L s^{3/2}\lambda.
\end{eqnarray*}
 Then clearly, $P\sup_{f\in\mathcal F} f^4 \leq (C s^{3/2}\lambda)^4.$
Then the condition \eqref{condi} is satisfied under the sparsity $s^3 (\log p)^{5/3} (\log n)^2/n=o(1).$\\
Or, for instance, for generalized linear models, if $w_\beta-w_{\beta_0}$ is uniformly bounded, 
then
\begin{equation}
\label{expr}
(P \sup_{f\in\mathcal F} f^4)^{1/4}  = (P\sup_{\beta: f_\beta\in\mathcal F} |\Theta_j^T x (w_\beta-w_{\beta_0})|^4)^{1/4}
\leq (P |\Theta_j^T x |^4 K)^{1/4}
.
\end{equation} 
Then under a moment condition $P |\Theta_j^T x |^4=\mathcal O(1)$, the expression \eqref{expr} is bounded. Hence we would require $s^3 (\log p)^{3} (\log n)^6/n=o(1).$

\noindent
\section{Examples}
\label{sec:examples}
\subsection{Quadratic loss }
\label{subsec:ex.lasso}
Consider the linear model 
\begin{equation} \label{lr2}
Y = X\beta^0 + \epsilon,
\end{equation}
where 
$\epsilon:=(\epsilon_1,\dots,\epsilon_n)$ with $\epsilon_i$'s independent and $\mathbb E\epsilon =0$. $X$ is a $n\times p$ matrix independent of $\epsilon$ with i.i.d. rows with mean zero and covariance matrix $\Sigma:=\mathbb E X_iX_i^T.$
We assume the inverse $\Theta = \Sigma^{-1}$ exists  and suppose that 
$\Lambda_{\max}(\Sigma)=\mathcal O(1)$ and 
$1/\Lambda_{\min}(\Sigma)=\mathcal O(1)$. Moreover, we assume that $\|X_i\|_\infty \leq K, \mathbb E|\Theta_j^T x|^4 = \mathcal O(1)$, 
$\|\beta_0\|_0 \leq s, \|\Theta_j\|_0\leq s$  and $\Lambda_{\max}(\hat \Sigma)=\mathcal O(1).$ Finally, assume the sparsity condition $s^3 (\log p)^2 (\log n)^2 /n = o(1)$.
\\
Consider the Lasso estimator 
\begin{equation}
\hat\beta_{\emph{}} := \text{arg}\min_{\beta\in\mathbb R^p} \frac{1}{n}\sum_{i=1}^n (Y_i-X_i^T \beta)^2 + \lambda \|\beta\|_1 .
\end{equation}
and its de-sparsified version $\hat b=\hat\beta + \hat\Theta(Y-X\hat\beta)/n.$
We apply Theorem  \ref{glm.main}.
The loss function is differentiable, so one may take 
$$\psi_\beta(x,y):=\dot \rho_\beta(x,y) = 2(y-x^T\beta)x.$$
By the above assumptions, the weight function is $w(u,y) = 2(y-u),$ and hence 
$\dot w(u,y)=-2$ is Lipschitz in $u$.
Hence conditions \ref{itm:bounded}, \ref{itm:initial.glm}, \ref{itm:sparsity.glm}, \ref{itm:w},  \ref{itm:wdot}, \ref{eig}, \ref{A3} are satisfied.

\subsection{Absolute loss}
\label{subsec:qre}
Consider the linear model \eqref{lr2}. 
The $\ell_1$-penalized least absolute deviations (LAD) estimator is defined by
\begin{equation}\label{qrp}
\hat\beta_{\text{LAD}} := \text{arg}\min_{\beta\in\mathbb R^p} \frac{1}{n}\sum_{i=1}^n |Y_i-X_i^T \beta| + \lambda \|\beta\|_1 .
\end{equation}
One may take
$$\psi_\beta(x,y) := \text{sign}(y -x^T \beta)x.$$
Lemma \ref{qr.ee} below shows that the estimating equations are approximately satisfied with $\psi_\beta$ at the point $\hat\beta_{\text{LAD}}$. 
We also need to construct an estimate of $\Theta = \frac{1}{2f_\epsilon(0)}(\mathbb E X_iX_i^T)^{-1}$. A near-oracle estimate of 
$(\mathbb E X_iX_i^T)^{-1}$ can be obtained using nodewise regression with input matrix $\hat\Sigma := X^T X/n$ under conditions \ref{bd.qr}, \ref{n1} below (see \cite{vdgeer13}).

\begin{enumerate}[label=(F\arabic*),start=1]
\item   
\label{bd.qr}
Assume that there exists a universal constant $K>0$ such that  $\|X_i\|_\infty \leq K.$
Moreover, $X_i$ are independent of $\epsilon_i,$ for $i=1,\dots,n.$
\item
\label{F}
Let the distribution function $F_\epsilon$ of $\epsilon_i$ satisfy $F_\epsilon(0)=1/2$ and let it have a density $f_\epsilon.$ Furthermore, $f_\epsilon(0) \geq c>0,$ where $c$ is a fixed constant and
$|f_\epsilon(0) |\leq K$  for all $x\in\mathbb R$ and a universal constant $K.$ Suppose that $f_\epsilon$ is Lipschitz.
\item
\label{n1}
(Nodewise regression)\\
Let $ \Theta':=(\mathbb E  X_iX_i^T)^{-1}$ and
suppose that 
${\Lambda_{\max}( \Theta')} =\mathcal O(1)$ and $1/\Lambda_{\min}( \Theta') =\mathcal O(1).$ Moreover, assume that
 $\mathbb E((\Theta_j')^T X_i)^4 = \mathcal O(1).$
\end{enumerate}

\begin{lemma}\label{qr.ee}
Assume the linear model \eqref{lr2}, suppose that condition  \ref{bd.qr} is satisfied and that $s/n=\mathcal O(\lambda)$. Let
$\psi_\beta(x,y) :=  \emph{sign}(y -x^T \beta)x.$
Then 
$$\|\mathbb P_n\psi_{\hat\beta_{\emph{LAD}}}\|_\infty = \mathcal O_P(\lambda).$$
\end{lemma}

\begin{lemma}\label{qr.entropy} 
Let
$$\mathcal F := \{\Theta_j^T(\psi_\beta(x,y)-\psi_{\beta_0}(x,y)):
\beta\in\mathbb R^p,
\|\beta\|_0 \leq s
\}.$$
Then for all $\epsilon>0$ it holds
$$\log  N(\varepsilon\|F\|_n, \mathcal F, \|\cdot \|_{n}) \leq C s\log p + 2Cs\log(16e/\varepsilon),$$
for some constant $C>0.$
\end{lemma}

\noindent
\noindent
Define $\sigma_j^2:=\frac{1}{4f_\epsilon(0)^2}\Theta_{jj},$ then we have the following result for the de-sparsified LAD estimator.
\begin{theorem}\label{qr.thm}
Assume the model \eqref{lr2} and suppose that conditions \ref{bd.qr}, \ref{F}, \ref{n1} are satisfied. Let $\hat \beta_{\text{LAD}}$ be defined in \eqref{qrp}. Let an estimate $\hat{\Theta}'$ of $\Theta'$ be constructed using nodewise regression with the input  matrix $\hat\Sigma:=X^TX/n,$
and let $\hat\Theta_j:= \hat{\Theta}'_j/(2f_\epsilon(0))$.
Let $\|\Theta_j\|_0 \leq s$ and $s^5 (\log p)^3 (\log n)^4 /n = o(1).$ 
Then for
$$\hat b_{\emph{LAD},j} := \hat \beta_{\emph{LAD,j}} +  \s \emph{sign}(Y_i -X_i^T \hat\beta_{\emph{LAD}})\hat\Theta_j^T X_i$$
it holds
$$\sqrt{n}(\hat b_{\emph{LAD},j}-\beta_j^0)/\sigma_j \rightsquigarrow \mathcal N(0,1).$$
\end{theorem}

\subsection{Huber loss}
\label{subsec:ex.huber}
We again consider the linear model.
The loss function is given by $\rho_\beta(x,y):= \rho(y-x^T\beta)$, where 
$$\rho(z) = [z^2 1_{|z|\leq K} + K (2|z| - K)1_{|z|>K}]/(2K),$$
for some constant $K>0.$
We  note that the first derivative satisfies the Lipschitz condition
$$|\dot \rho (u,y) - \dot \rho(u',y)| \leq |u'-u|, \text{ for all }u,u',y.$$
Hence we may apply the first part of  Theorem \ref{glm.main}. Define the $\ell_1$-penalized Huber estimator
\begin{equation}\label{hub}
\hat\beta_{\text{Huber}} := \text{arg}\min_{\beta\in\mathbb R^p} \frac{1}{n}\sum_{i=1}^n \rho_{\text{Huber}}(Y_i-X_i^T \beta) + \lambda \|\beta\|_1 .
\end{equation}
Define the function
\[
\psi_\beta(x_i,y_i)=\psi(y_i-x_i\beta):=\psi(z)= 
\begin{cases}
z/K &\text{ if } |z|\leq K, \\
\text{sign}(z) & \text{ if }|z|  > K.
\end{cases}
\]
For Huber loss, we have 
$$(\mathbb E w_{\text{Huber},j}(y,x^T\beta))'_{\beta=\beta_0} = (F_{\epsilon}(K)-F_\epsilon(-K))/K,$$
and hence
$$\Sigma=(\mathbb E \psi(y,x^T\beta))'_{\beta=\beta_0} =({F_{\epsilon}(K)-F_\epsilon(-K)})/K\mathbb E xx^T.$$ 
Furthermore,
$$\mathbb E \psi_{\beta_0}\psi_{\beta_0}^T = \mathbb E w_{\text{Huber},j}^2(y,x^T\beta_0) xx^T=
\left\{
\frac{1}{K}\mathbb E_\epsilon 1_{|\epsilon|<K}\epsilon^2 
+ [F_\epsilon(-K) + 1-F_\epsilon(K)]\right\}\mathbb Exx^T.$$
Hence the asymptotic variance  per entry of the de-sparsified estimator is 
$$\sigma_{\text{Huber},j}^2:= 
\Theta_j^T \mathbb E \psi_\beta\psi_\beta^T\Theta_j 
= K^2\frac{\mathbb E_\epsilon 1_{|\epsilon|<K}\epsilon^2/K 
+ [F_\epsilon(-K) + 1-F_\epsilon(K)]}{(F_{\epsilon}(K)-F_\epsilon(-K))^2}\Theta_{jj}^0.$$
One could then define an estimator of $\sigma_{\text{Huber},j}^2$ as follows
$$\hat\sigma_{\text{Huber},j}^2 := 
\frac{1}{K^2}\frac{\frac{1}{K}\s 1_{|\hat\epsilon_i|<K}\hat\epsilon_i^2 
+ \s 1_{|\hat\epsilon_i| > -K} }{(\s 1_{|\hat\epsilon_i| < -K})^2}\hat\Theta_{jj},$$
where $\hat \epsilon_i:=Y_i-X_i^T\hat\beta_{\text{Huber}}$ and  $\hat \Theta_{jj}$ was obtained using nodewise regression with matrix $X$. 
\begin{theorem}
Assume the model \eqref{lr2} and suppose that conditions \ref{bd.qr}, \ref{n1} are satisfied. Let 
$\hat\beta_{\emph{Huber}}$ be defined in \eqref{hub}.
Let an estimate $\hat{\Theta}'$ of $ \Theta'$ be constructed using nodewise regression with the input  matrix $\hat\Sigma:=X^TX/n,$
and let $\hat\Theta_j:= \hat{\Theta}'_j/( F(K)- F(-K))$.
Let $\|\Theta_j\|_0 \leq s$ and $s^3 (\log p)^2 (\log n)^2 /n = o(1).$ 
Then for
$$\hat b_{\emph{Huber},j}:= \hat\beta_{\emph{Huber},j} + \frac{1}{n}\sum_{i=1}^n\hat\Theta_j^T \psi_{\hat\beta}(X_i,Y_i)
$$
it holds
$$\sqrt{n}(\hat b_{\emph{Huber},j}-\beta^0_j)/\sigma_{\emph{Huber},j} \rightsquigarrow \mathcal N(0,1).$$
\end{theorem}

\section{Conclusions}
\label{sec:conc}
\par
A first message of our analysis is that to obtain asymptotically normal estimator in high-dimensional settings, we do not require the loss to be twice differentiable. 
Instead however, one must assume that the expected loss is sufficiently smooth and that the class of the ``score'' functions indexed by the unknown parameter satisfies a certain entropy bound. To this end, our analysis needed sparsity in the Lasso estimator, which was shown in Section \ref{sec:glm.main}.
A second message is that we need to estimate the score of the nuisance parameter for the  methodology to work. 
There is  a certain price we pay compared to the results in \cite{vdgeer13}:
our analysis leads to somewhat stronger sparsity assumptions. The analysis in \cite{vdgeer13} requires a sparsity condition
 $s=o(\sqrt{n}/\log p)$ for generalized linear models, where $s=\max\{\|\beta_0\|_0,\|\Theta_j\|_0\}$. We need the somewhat stronger condition $s^{3/2}=o(\sqrt{n}/\log p)$. This results from considering non-differentiable loss functions on one hand, and from estimation of the score for the nuisance parameter on the other.

\section{Simulation study}
\label{sec:sims}
We confirm our theoretical results with numerical experiments on synthetic data  and compare the performance of our approach with other plausible procedures, such as the maximum likelihood estimator.
\subsection{Models}
We consider the linear model with a continuous random variable 
and logistic  regression
with a binary response variable. In both settings, the design matrix $X\in\mathbb R^{n\times p}$ has independent normally distributed rows with $\mathbb EX=0$ and with covariance matrix $\Sigma_0:=\Theta_0^{-1}$ where the precision matrix is given by
$$\Theta_{ij}^0 = 0.3 \text{ if } |i-j|=1, \Theta_{ii}^0=1 \text{ and }\Theta_{ij}^0 = 0\text{ otherwise}.$$
The vector of regression coefficients is sparse and is given by $\beta_0=(1,1,1,0,\dots,0)\in\mathbb R^p.$ 

\subsection{Confidence intervals and hypothesis testing}
Our proposed methodology gives us tools to construct confidence intervals and test hypothesis about the regression parameters. 
We can also apply multiple testing procedures to test hypotheses about sets of regression coefficients.
\par 
We construct confidence intervals using the asymptotic normality of the de-sparsified Lasso $\hat b$. 
In particular, an asymptotic  $(1-\alpha)\%$ confidence interval for $\beta_j^0$ ($j=1,\dots,p$) can be constructed by
$$\hat b_j \pm \Phi^{-1}(1-\alpha/2)\hat\sigma_j/\sqrt{n},$$
where $\hat\sigma_j$ is an estimate of the asymptotic variance of the de-sparsified estimator (see Lemma \ref{var} and Section \ref{sec:examples}).
To calculate the de-sparsified estimator, we first need to compute the initial Lasso estimator which is done using the function {\small\texttt{glmnet()}} (or {\small\texttt{cv.glmnet()}}) from the R package {\small\texttt{glmnet}}. 
The matrix $\Theta$ is estimated using the nodewise Lasso regression.
The de-sparsified estimator is calculated as in \eqref{desp}.
Asymptotic variance of the de-sparsified estimator is estimated as in Lemma \ref{var}. 
\par
For the confidence intervals, we report average coverages and averages lengths from $N$ independently generated samples.
We give the average coverage over the ``active'' set $S_0=\{j:\beta_j^0\not = 0\}$ and average coverage over the ``non-active'' set $S_0^c$.
Similarly, we report average lengths of the confidence intervals over $S_0$ and $S_0^c$.
\par
For testing multiple hypothesis such as $H_0: \beta_j^0 = 0$ among all $j=1,\dots,p,$ we will use Bonferroni-Holm multiple testing adjustment to control the family-wise error rate (FWER) or the Benjamini-Hochberg correction for controlling the false discovery rate (FDR) (see \cite{BH}). 

\subsection{Linear regression}
In this section, we investigate the performance of the de-sparsified Lasso estimator with different loss functions (square loss, absolute loss and Huber loss) on simulated data. 
We consider the linear regression setting 
$$Y=X\beta_0+\epsilon,$$
with independent errors, which are independent of the design matrix and
 have \\
(1) a Gaussian distribution with zero mean and variance one, or\\
(2) Student $t_5$-distribution (scaled to have variance equal to one), or\\
(3) Student $t_3$-distribution (scaled to have variance equal to one).
\par
The construction of the estimators is the same as in Section \ref{sec:examples}. The tuning parameters were selected by cross-validation for square loss and Huber loss (using R packages \texttt{glmnet} and \texttt{hqreg}) and for absolute loss, the tuning parameter was selected by the method \texttt{qr.fit.lasso} from the R package \texttt{quantreg}. 
We assume that the weights needed to calculate the de-sparsified LAD estimator and de-sparsified Huber estimator (and their asymptotic variances) are known. Their estimation would involve e.g. density estimation and deeper analysis is omitted in this paper.  
We report the results on confidence intervals in Table \ref{tab:cover} and the histograms for the three methods in Figure \ref{fig:hist} (for the case of Gaussian error). 

\begin{table}
\caption{\label{tab:cover}
\footnotesize A table showing a comparison 
of the de-sparsified Lasso (D-S Lasso), the de-sparsified LAD estimator (D-S $\ell_1$-LAD) and the de-sparsified Huber estimator with $K=0.5$ (D-S Huber ($K=0.5$)).
Here, $\beta_0=(1,1,1,0,\dots,0).$
 The number of generated random samples was $N=100.$ The nominal coverage is 0.95.
}
\centering
\fbox{
\begin{tabular}{cclcccc}
   \multicolumn{3}{c}{\multirow{2}{*}{\bf Gaussian-distributed errors}}        &  &  &  &  \\
	\multicolumn{2}{c}{}     &  &\multicolumn{2}{c}{Coverage} & \multicolumn{2}{c}{Length}   \\  
	$p$& {  $n$ }      &  & $S_0$ & $S_0^c$ & $S_0$ & $S_0^c$  \\
\hline\\[-2ex]
\multirow{3}{*}{$100$} &\multirow{3}{*}{ $500$} 
& D-S Lasso & 92.67 & 95.88 & 0.17 & 0.17\\
&& D-S $\ell_1$LAD &  90.00 & 91.37 & 0.22 & 0.22 \\ 
&& D-S Huber ($K=0.5$) & 94.33 & 95.62 & 0.19 & 0.19
\\[0.07cm]\specialrule{.1em}{.05em}{.05em}  \\[-2ex]
   \multicolumn{3}{c}{\multirow{2}{*}{\bf Student $t_3$-distributed errors}}        &  &  &  &  \\
	\multicolumn{2}{c}{}     &  &\multicolumn{2}{c}{Coverage} & \multicolumn{2}{c}{Length}   \\  
	$p$& {  $n$ }      &  & $S_0$ & $S_0^c$ & $S_0$ & $S_0^c$  \\
\hline\\[-2ex]
\multirow{3}{*}{$100$} &\multirow{3}{*}{ $500$} 
& D-S Lasso & 95.00 & 95.94 & 0.17 & 0.17 \\
&& D-S $\ell_1$LAD&  91.33 & 88.19 & 0.13 & 0.13  \\ 
&& D-S Huber ($K=0.5$) & 94.67 & 95.51 & 0.11 & 0.11
\\[0.07cm]\specialrule{.1em}{.05em}{.05em}  \\[-2ex]
   \multicolumn{3}{c}{\multirow{2}{*}{\bf Student $t_5$-distributed errors}}        &  &  &  &  \\
	\multicolumn{2}{c}{}     &  &\multicolumn{2}{c}{Coverage} & \multicolumn{2}{c}{Length}   \\  
	$p$& {  $n$ }      &  & $S_0$ & $S_0^c$ & $S_0$ & $S_0^c$  \\
\hline\\[-2ex]
\multirow{3}{*}{$100$} &\multirow{3}{*}{ $500$} 
& D-S Lasso & 92.00 & 95.79 & 0.17 & 0.17 \\
&& D-S $\ell_1$LAD&  89.33 & 89.04 & 0.18 & 0.17  \\ 
&& D-S Huber ($K=0.5$) & 91.00   & 95.62  & 0.15 & 0.15
\end{tabular}
}
\end{table}

\begin{figure}
\centering
\begin{center}
\bf \small 
Asymptotic normality of regression parameters
\end{center}
\vskip 0.3cm

De-sparsified Lasso

\includegraphics[width=0.5\textwidth]{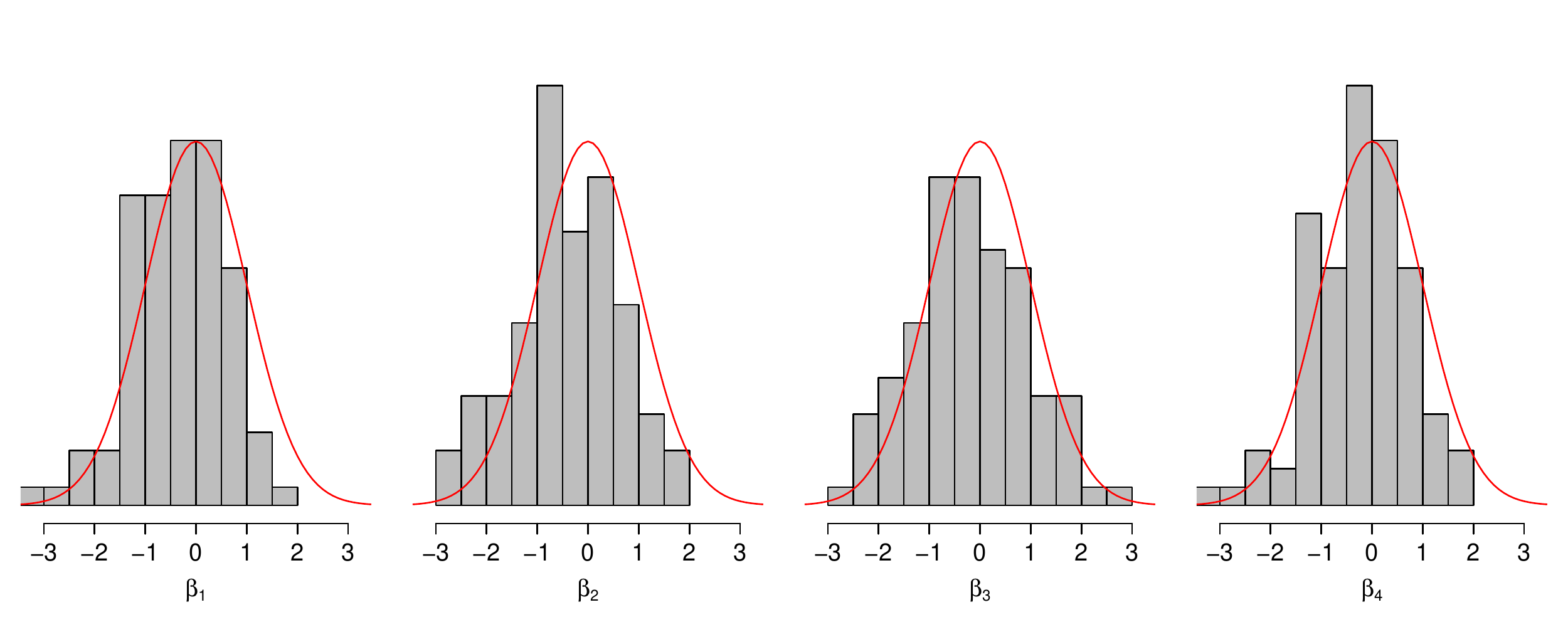}

De-sparsified LAD 

\includegraphics[width=0.5\textwidth]{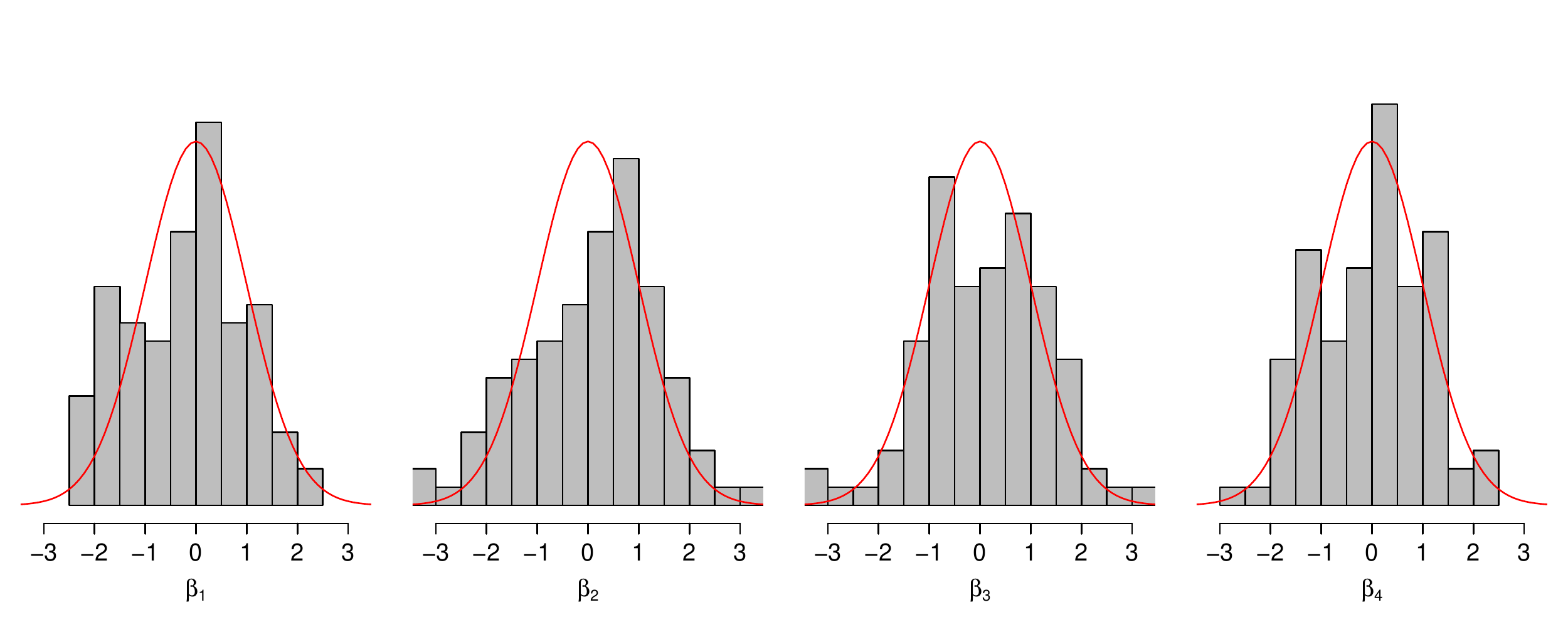}

De-sparsified Huber estimator

\includegraphics[width=0.5\textwidth]{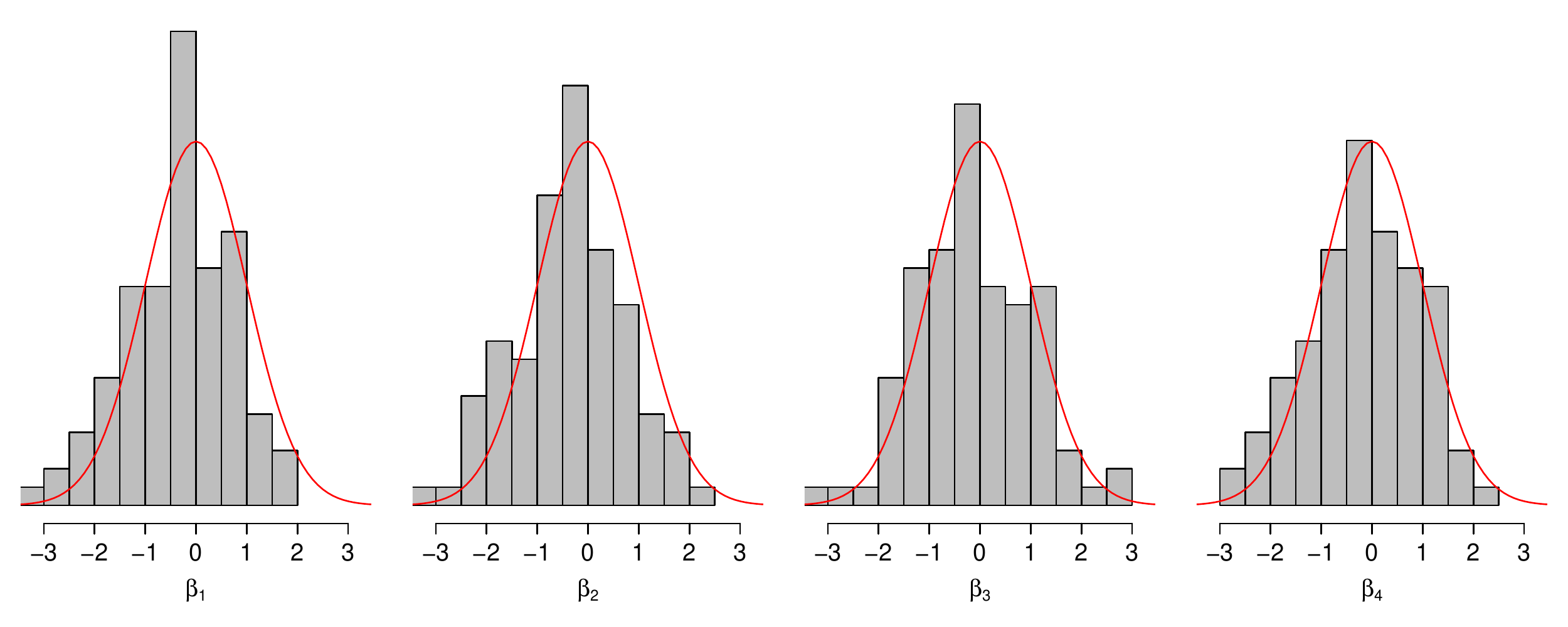}
\caption{Histograms of $\sqrt{n}(\hat b_j- \beta_j^0)/\hat\sigma_j,j=1,\dots,4$ for the de-sparisfied LAD, de-sparsified Lasso, and de-sparsified Huber estimator ($K=0.5$). Here, $n=500,p=100,\beta_0=(1,1,1,0,\dots,0).$
The error distribution is $\mathcal N(0,1)$.
Superimposed is the density of $\mathcal N(0,1)$ (red curve).
}
\label{fig:hist}
\end{figure}

\subsection{Logistic regression}

We analyze the performance of the de-sparsified Lasso for the case of logistic regression,
$$\log \left(\frac{\pi (x)}{1-\pi (x)}\right)=x^T\beta_0, \quad \pi(x) = P(Y=1|X=x).$$
 The $\ell_1$-penalized estimator is defined via the logistic loss function 
$$\rho_\beta(x,y) = -yx^T\beta + \log (1+e^{x^T\beta}).$$
\par
In the first part of the simulation experiment, we construct confidence intervals using the de-sparsified logistic Lasso as defined in the general formula \eqref{desp}. We consider a setting with $p=100$ and $n=400$. This also allows us to compare our approach to confidence intervals based on a maximum likelihood estimator. The maximum likelihood estimator is fitted with the function {\small\texttt{glm()}} in R. The confidence intervals are then calculated using the function {\small\texttt{confint.default()}}, which bases the confidence intervals on the standard error.
 The initial logistic Lasso estimator is fitted with {\small\texttt{cv.glmnet()}} with tuning parameter chosen by cross-validation. The matrix $\Theta$ is estimated by nodewise regression and the tuning parameters chosen by cross-validation.
\par
We plot histograms for the individual entries of the de-sparsified Lasso estimator in Figure \ref{fig:hist_logit}. For comparison, we also display histograms for the initial Lasso estimator and the maximum likelihood estimator. 
This demonstrates that the ``de-sparsifying'' is useful even in the setting where $p<n.$ For the confidence intervals, we report  average coverages and lengths over the active and non-active set in Table \ref{tab:cov}. The de-sparsified estimator performs significantly  better than the maximum likelihood estimator.
\par
In the second part of the experiment, we look at multiple testing. 
We consider testing the hypothesis $H_0:\beta_j^0=0$ among all $j=1,\dots,p.$ We use the Bonferroni-Holm procedure for controlling FWER. From 200 generated samples, 
the testing procedure had 100\% true positive rate and FWER value 0.015. 
\par 

\noindent
\begin{figure}
\small\centering
\textbf{\small \bf Histograms for coefficients in logistic regression}
\vskip 0.3cm
\begin{tabular}{ccc}
{\text{De-sparsified logistic Lasso} } &
 { \text{Logistic Lasso}} & \text{MLE} \\\cmidrule(lr){1-1}\cmidrule(rl){2-2}\cmidrule(rl){3-3}
\\[-2px]
\includegraphics[width=0.31\textwidth]{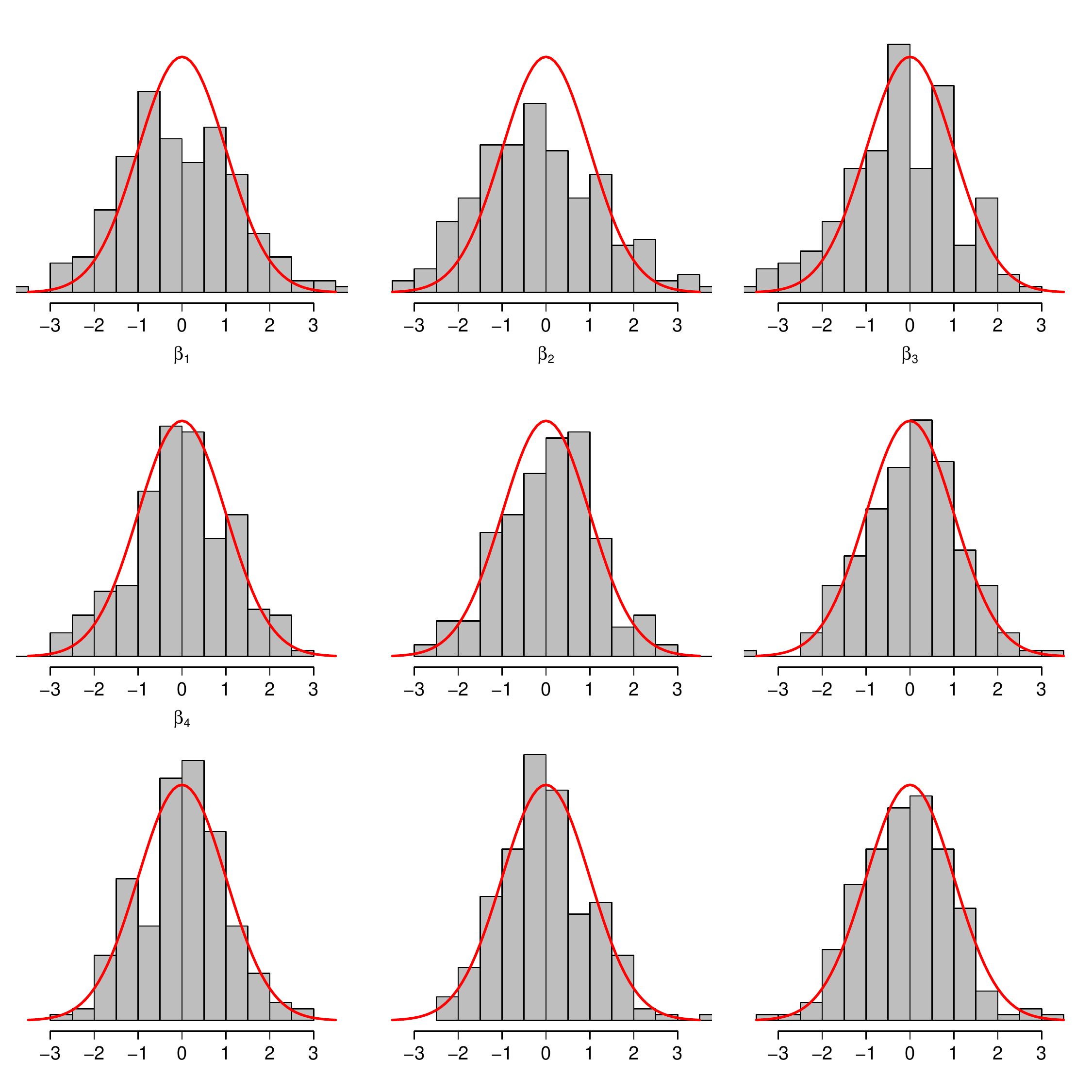} &
\includegraphics[width=0.31\textwidth]{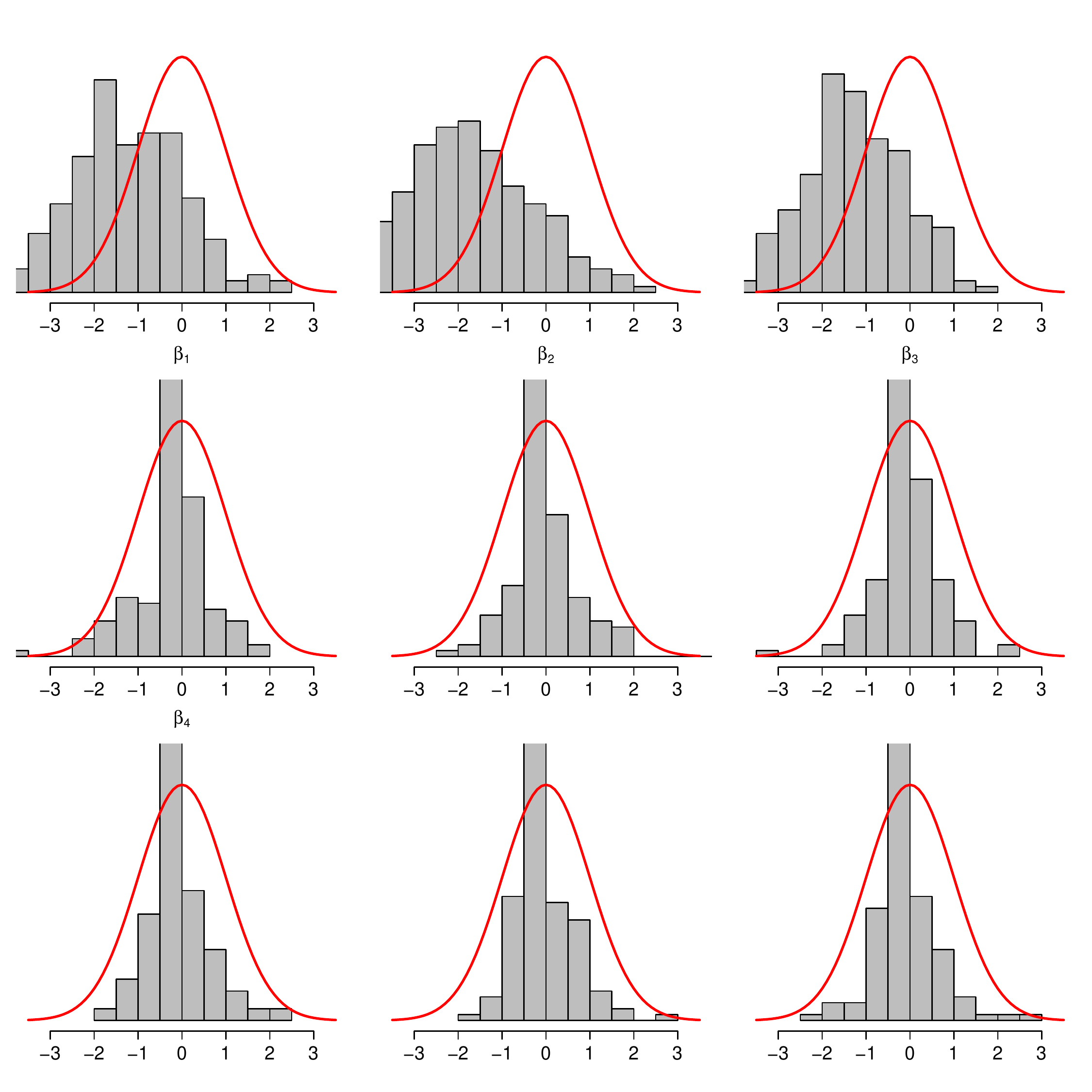} &
\includegraphics[width=0.31\textwidth]{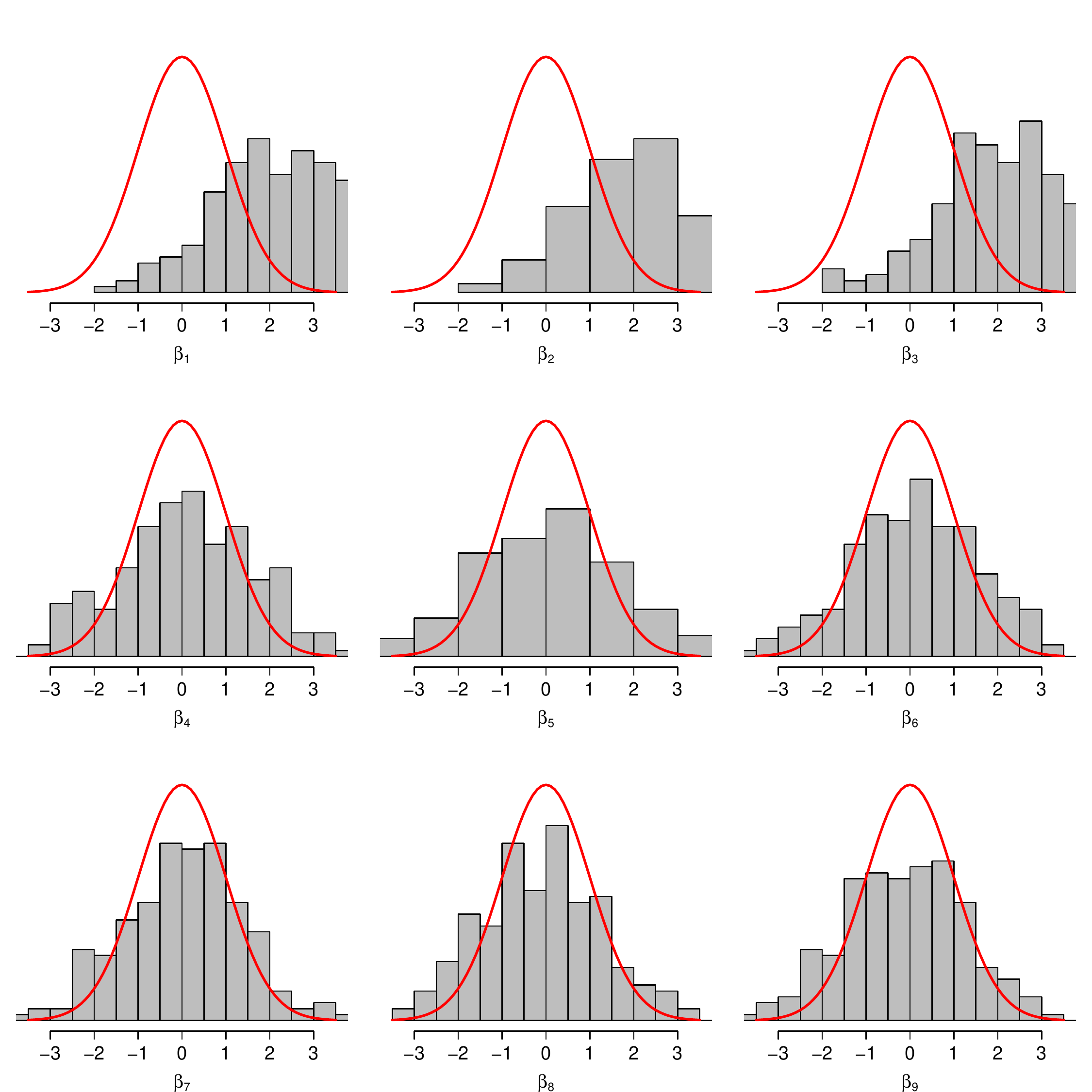}
\end{tabular}
\caption{Histograms of the de-sparsified logistic Lasso (left panel) for $\beta_1,\dots,\beta_9$. For comparison, histograms of the logistic  Lasso (right panel) are also displayed. Here, $n=800,p=100.$ Even for low-dimensional settings, the de-sparsifying step turns out to be useful.}
\label{fig:hist_logit}
\end{figure}

\begin{table} 
\caption{\label{tab:cov} 
\footnotesize
A table showing the average coverages and lengths over the active and non-active set for the de-sparsified logistic Lasso
(D-S Logistic Lasso) and the maximum likelihood estimator (MLE).
Here, $\beta_0=(1,1,1,0,\dots,0),p=100.$
}

\centering
\fbox{
\footnotesize
\begin{tabular}{clcccc}
 \multicolumn{2}{c}{\multirow{2}{*}{Logistic regression}} & \multicolumn{2}{c}{Coverage} & \multicolumn{2}{c}{Length} \\[1ex]
  & & $S_0$  & $S_0^c$     & $S_0$ & $S_0^c$  \\
\hline\\[0ex]
\multirow{2}{*}{$n=400$} & D-S Logistic Lasso &   0.817 & 0.919 & 0.423  &0.402 \\
  & MLE &   0.320 & 0.891 & 0.730  &0.638 \\[1ex]
\multirow{2}{*}{$n=800$}& D-S Logistic Lasso &   0.872 & 0.932 & 0.464  &0.374 \\
& MLE &   0.657 & 0.929 & 0.433  &0.348 \\
\\[0.07cm]
\end{tabular}
}

\end{table}

\newpage

\section{Real data experiments}
\label{sec:real}

In this section we investigate the practical usefulness of our methodology for gene expression studies which involve high-dimensional data.

\subsection{Linear regression: Riboflavin (vitamin B2) production}
\label{subsec:ribo}
The dataset {\small\texttt{riboflavin}} from the R package {\small \texttt{hdi}}  contains gene expression levels of 4088 genes and the  response variable represents riboflavin (vitamin B2) production. Our goal is to identify genes that significantly effect the production of riboflavin.
This dataset is ultra-high-dimensional given that it contains 4088 variables and only 72 observations, but we will reduce it to a moderate high-dimensional data set as for testing such a large number of hypotheses simultaneously turns out to be very conservative.
This was also demonstrated  in the papers \cite{vdgeer13} and \cite{jm}, which previously studied this data set. The paper 
\cite{vdgeer13} did not select any gene using the de-sparsified Lasso and the procedure suggested in \cite{jm} selected only two genes: genes YXLD\_at and YXLE\_at. 
The works apply (a version) of the de-sparsified Lasso with square loss to select significant variables using a multiple testing adjustment. We also aim to apply the de-sparsified Lasso estimator but in addition we apply the de-sparsified LAD estimator which is expected to be more robust to outliers and to the violation of the normality assumption.
\par
To do initial variable screening, we calculate 
$$\omega_i := |Y^TX_i|,\;\;i=1,\dots,4088,$$
where $X_i$ is the $i$-th row of the design matrix and $Y$ is the response. 
We select the first 300 variables which have the highest $\omega_i$'s.
To calculate the de-sparsified estimator, we fit the initial Lasso estimators with square loss and absolute loss to the data using cross-validation to choose the tuning parameters. To calculate an estimate of $\Theta$, we use nodewise square-root Lasso from Remark \ref{sqrtnode}, which avoids the need to do  cross-validation to choose the tuning parameters. 
 We then test the hypotheses: $H_0:\beta_j^0=0,$ among all $j=1,\dots,300.$ For multiple testing adjustment, we use two different procedures: the Bonferroni-Holm procedure and the Benjamini-Hochberg procedure.
The results are reported in Table \ref{tab:riboscreen}. 

\begin{center}
\begin{table}
\begin{center}
	Bonferroni-Holm	adjustment\\
	\fbox{
\begin{tabular}{llllll}
  \multicolumn{2}{l}{D-S LAD}  & \multicolumn{2}{l}{D-S LASSO} \\ 
	\hline
 RPSB\_at & 0.05 & RPLX\_at & 0.12 \\ 
   YCEI\_at & 0.14 & YCEI\_at & 0.41 \\ 
    & & YHCL\_at & 0.51 \\ 
    &  & YNEF\_at & 0.57 \\ 
    &  & YFKN\_at & 0.61 \\ 
\end{tabular}
}
\\[0.3cm]
Benjamini-Hochberg adjustment
\\
	\fbox{
\begin{tabular}{lllllllll}
  \multicolumn{2}{l}{D-S LAD}  & \multicolumn{2}{l}{D-S LASSO} \\ 
  \hline
 RPSB\_at & 0.05 & IOLI\_at & 0.10 \\ 
 YCEI\_at & 0.07 & RPLX\_at & 0.10 \\ 
 RPLX\_at & 0.36 & YCEI\_at & 0.10 \\ 
 YHCL\_at & 0.36 & YFKN\_at & 0.10 \\ 
 NARI\_at & 0.49 & YHCL\_at & 0.10 \\ 
\end{tabular}
}
\end{center}

\caption{\label{tab:riboscreen}
  Variables with smallest p-values among all genes selected by initial screening are reported. D-S LAD corresponds to the de-sparsified Lasso and D-S LASSO corresponds to the de-sparsified LAD estimator. Corresponding $p$-values are reported next to the genes. }
\end{table}

\end{center}

\subsection{Logistic regression: Genome-wide studies in cancer
}
We apply our methodology to a real data set on genome-wide association studies in cancer. The response variable indicating presence or absence of the illness (prostate cancer) is binary, therefore we model the relationship using logistic regression. 
The dataset contains 102 observations (52 positive, 50 control) on 6033 genes and is available from the R package {\small \texttt{spls}}.
\par
We do variable screening as in Section \ref{subsec:ribo} to reduce the ultra-high-dimensional data to a more feasible size of 200 genes.
The initial logistic Lasso estimator is computed using {\small\texttt{cv.glmnet()}} with cross-validation to determine the tuning parameter. The nodewise regression estimator of $\Theta$ is computer using the square-root Lasso as in Remark \ref{sqrtnode}.
\par 
%

 Using the de-sparsified logistic Lasso and multiple testing adjustment (both Bonferroni-Holm and Benjamini-Hochberg yield the same result), we identify gene number 515 as significant, with a coefficient estimate $\hat b_{515}=-2.4677139.$
By thresholding the de-sparsified logistic Lasso at the level $2\hat\sigma_j\sqrt{\log p/n},$ $j=1,\dots,p$ we find genes 515, 4639, 5503 significant with coefficients -2.4677139, -1.3019043,  -0.7832844, respectively.
For a comparison, logistic Lasso identifies 32 genes with non-zero coefficients (including genes 515, 4639, 5503). 
\par For an illustration of the confidence intervals for individual coefficients (without adjustment), see Figure \ref{fig:cipro}.

\begin{figure}
\centering
\includegraphics[width=0.7\textwidth]{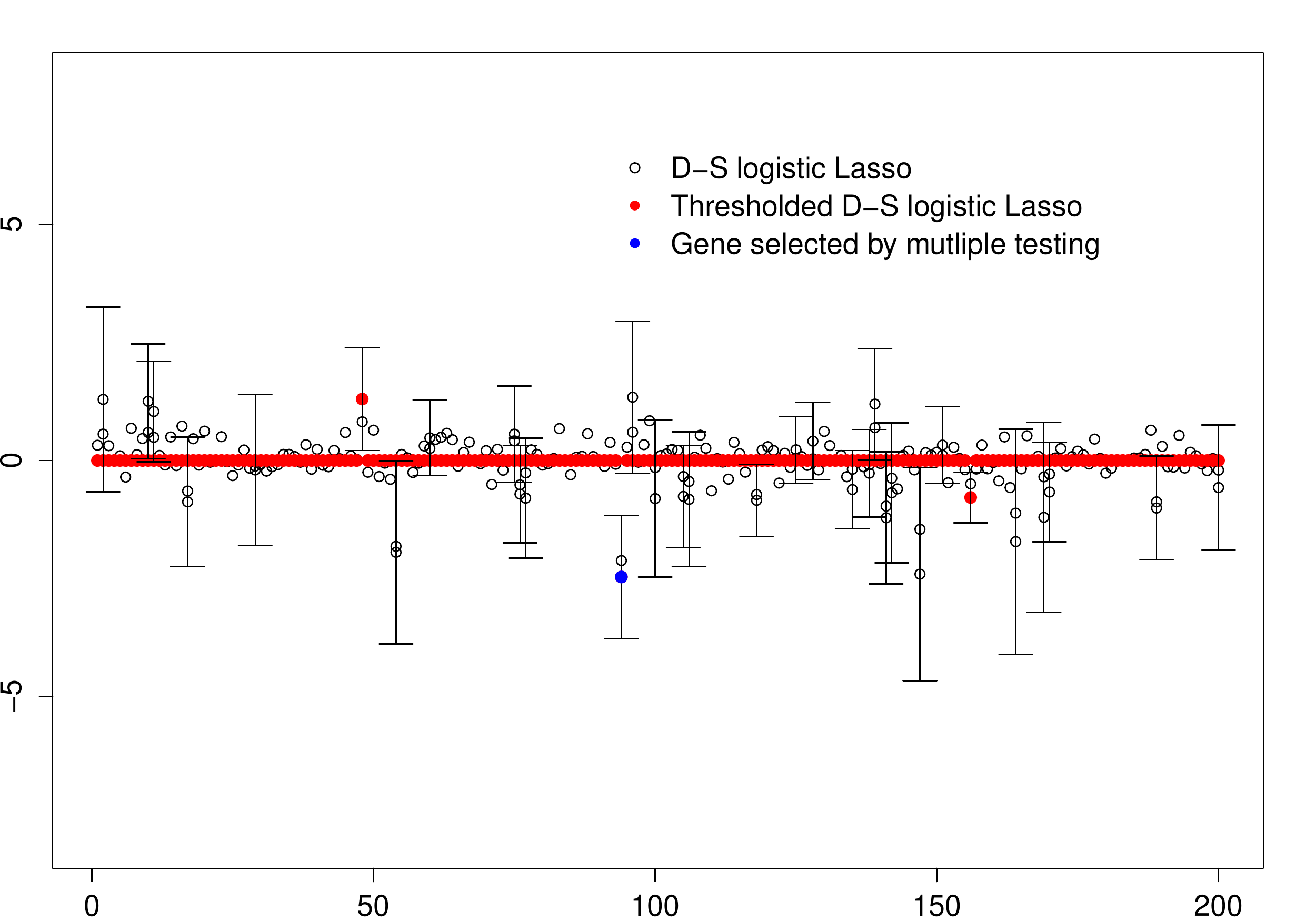} 
\caption{Prostate cancer dataset: The vertical segments represent individual confidence intervals (without adjustment). For clarity of presentation, only a fraction of all the confidence intervals is displayed. The red dots represent variables that were selected by thresholding and the blue dot represents a variable selected by multiple testing. }
\label{fig:cipro}
\end{figure}

\subsection{Discussion}
The simulation study demonstrated that the de-sparsified estimator performs well in a variety of settings for the linear regression and logistic regression, in the setting when $p$ is moderately large. In these settings, the de-sparsified estimator proves to be useful as it clearly outperforms the maximum likelihood estimator.  
We further observed that multiple testing with the de-sparsified estimator turned out to be conservative and lead to only a few variables selected. However, this is to be expected in the presence of many variables.

\bibliography{../../../gminf}

\appendix
\newpage

\section{Appendix}
\label{sec:proofs}
The section is organized as follows.
\begin{enumerate}
\item
Several preliminary results are stated in Section \ref{subsec:rates.ep}.
\item
Proofs for Section \ref{sec:glm.main} (generalized linear models) can be found in Section \ref{sec:glm.proof}.
\item
Proofs for Section \ref{sec:main.all} (general models) are contained in Section \ref{sec:main.all.proof}.
\item
Proofs for Section \ref{sec:nw} (nodewise regression) are given in Section \ref{sec:nw.proof}.
\item
Proofs for Section \ref{sec:examples} (examples) are contained in Section \ref{sec:examples.proof}.
\end{enumerate}
For two sequences, $f_n,g_n$, we write $f_n \lesssim g_n$ if there exists a constant $C>0$ such that $f_n\leq Cg_n$ for all $n.$
\subsection{Preliminary material}
\label{subsec:rates.ep}

\label{subsec:suff}
We define
$$Z(\mathcal F) := \sup_{f\in\mathcal F} |(\mathbb P_n-P)f|,$$
and
$$Z^\epsilon(F) := \sup_{f\in\mathcal F} |\sum_{i=1}^n \epsilon_i f(X_i)|.$$


\begin{theorem}[see e.g. \cite{vdv}]
\label{sym}
$$\mathbb EZ(\mathcal F) \leq 2 \mathbb EZ^\epsilon (\mathcal F).$$
\end{theorem}

\begin{theorem}[Dudley's inequality]\label{dudley}
$$\mathbb EZ^\epsilon(\mathcal F) \leq C_0 
\inf_{\delta>0} \mathbb E\left[ \hat R \int_\delta^1 \sqrt{\log N(u\hat R, \mathcal F, \|\cdot\|_n)}du /\sqrt{n} + \delta \hat R\right],$$
where $\hat R:= \sup_{f\in\mathcal F} \|f\|_n.$
\end{theorem}


For reader's convenience, we recall the Nemirovski inequality.

\begin{theorem}[Nemirovski inequality, see e.g. \cite{hds}] \label{nemirovski}
For $m\geq  1$ and $p \geq e^{m-1}$, we have 
$$\mathbb E \max_{1\leq j \leq p} |\sum_{i=1}^n \gamma_j(Z_i) - \mathbb E\gamma_j(Z_i)|^m \leq {(8\log (2p))^{m/2}} 
\mathbb E \left(\max_{1 \leq j \leq p} \sum_{i=1}^n \gamma_j^2(Z_i)\right)^{m/2}.$$
\end{theorem}

\subsection{Proofs for Section \ref{sec:main.all} (General high-dimensional models) }
\label{sec:main.all.proof}
\begin{proof}[Proof of Theorem \ref{mai}]
Consider the decomposition 
\begin{eqnarray*}
\hat\beta_j-\beta_j^0 -\Theta_j^T \mathbb P_n \psi_{\hat\beta} 
&=&
- \Theta_j^T\mathbb  P_n \psi_{\beta_0}\\
&&
-\underbrace{\Theta_j^T(\mathbb P_n-P)(\psi_{\hat\beta} - \psi_{\beta_0} )}_{i} +
\underbrace{\hat\beta_j-\beta_j^0 -\Theta_j^T P(\psi_{\hat\beta} - \psi_{\beta_0}) }_{ii}.
\end{eqnarray*}
By assumption, we have $i=\Theta_j^T(\mathbb P_n-P)(\psi_{\hat\beta} - \psi_{\beta_0} )=o_P(1/\sqrt{n})$.
Next we treat the term $ii.$
By condition \ref{itm:eig0}, we have $\|\Theta_j\|_1 \leq \sqrt{s} \mathcal O(1).$
Condition \ref{taylor} and assumption $d^2(\hat\beta,\beta_0)=o_P\left(\frac{1}{\sqrt{ns}}\right)$ then
 yield
$$ii =\hat\beta_j-\beta_j^0 - \Theta_j^T P(\psi_{\hat\beta} - \psi_{\beta_0}) = \mathcal O(\|\Theta_j\|_1 d^2(\hat\beta,\beta_0)) 
= \mathcal O_P(\sqrt{s} d^2(\hat\beta,\beta_0)) = o_P(1/\sqrt{n}).$$
Hence 
we conclude
\begin{equation}\label{betahat}
\hat\beta_j - \beta^0_j -\Theta_j^T\mathbb P_n\psi_{\hat\beta}= 
- \Theta_j^T \mathbb P_n\psi_{\beta_0}  + o_P(1/\sqrt{n}).
\end{equation}
The sparsity condition \eqref{spa.prop} and condition \ref{itm:initial} imply
$$\|(\hat\Theta_j - \Theta_j)^T \mathbb P_n\psi_{\hat\beta}\|_\infty \leq \|\hat\Theta_j - \Theta_j\|_1 \| \mathbb P_n\psi_{\hat\beta} \|_\infty =\Op (s\lambda^2)=o_P(1/\sqrt{n}).$$
Hence \eqref{betahat} holds with $\Theta_j^T\mathbb P_n\psi_{\hat\beta}$ replaced by $\hat\Theta_j^T\mathbb P_n\psi_{\hat\beta}$.
That means
\begin{equation}
\hat\beta_j - \beta^0_j -\hat\Theta_j^T\mathbb P_n\psi_{\hat\beta}= 
- \Theta_j^T \mathbb P_n\psi_{\beta_0}  + o_P(1/\sqrt{n}).
\end{equation}
By Condition \ref{itm:norm}, the term $ \Theta_j^T \mathbb P_n\psi_{\beta_0}$
is asymptotically normal $\mathcal N(0,1)$ when normalized by the square-root of variance 
\begin{eqnarray*}
\text{Var}(\Theta_j^T \mathbb P_n\psi_{\beta_0})
=
 \Theta_j^T P\psi_{\beta_0}\psi_{\beta_0}^T \Theta_j/n.
\end{eqnarray*}
The result follows.
\end{proof}

\begin{proof}[Proof of Theorem \ref{main2}]
By Theorems \ref{sym} and \ref{dudley} (upper-bounding the entropy  integral by taking its lower bound equal to zero), it follows that
\begin{eqnarray*}
\mathbb E \sup_{f\in\mathcal F} |(\mathbb P_n-P)f| 
&\leq&
 2C_0 
 \mathbb E\left[  \int_0^{\hat R} \sqrt{\log N(u, \mathcal F, \|\cdot\|_n)}du /\sqrt{n} \right],
\end{eqnarray*}
where $\hat R:=\sup_{f\in\mathcal F} \|f\|_n.$
Let $F(x)=\sup_{f\in\mathcal F} |f(x)|$ be the envelope function of $\mathcal F$. Note that $\hat R \leq \|F\|_n$ by the definition of the envelope function $F$. Moreover, note that $N(u, \mathcal F, \|\cdot\|_n) \leq N(u/4, \mathcal F, \|\cdot\|_n)$.
Thus, and by  the assumed entropy condition, we obtain
\begin{eqnarray*}
\mathbb E \sup_{f\in\mathcal F} |(\mathbb P_n-P)f| 
&\leq &
 2C_0 
\mathbb E\left[  \int_0^{\hat R} \sqrt{\log N(u/4, \mathcal F, \|\cdot\|_n)}du /\sqrt{n} \right]
\\
&\leq & 2C_0
 \mathbb E\int_{0}^{\hat R}\sqrt{s\log p + s\log (4\|F\|_n/u)}du /\sqrt{n}  
\\
&\leq &
 2C_0 [\mathbb E {\hat R}\sqrt{s\log p/n} + \mathbb E\int_{0}^{\hat R}\sqrt{ s\log (4\|F\|_n/u)}du /\sqrt{n}  ],
\end{eqnarray*}
where we also used that for any $x,y>0$ it holds that $\sqrt{x+y} \leq \sqrt{x}+\sqrt{y}.$
One can show the following upper bound using integration per partes:
$$\int_0^\delta \log (1/u)du \leq \delta\log \frac{1}{\delta} \frac{1}{1-\frac{1}{\log \frac{1}{\delta}}}.$$
Then for $\delta \leq 1/4$ it holds that $\frac{1}{1-\frac{1}{\log \frac{1}{\delta}}} \leq \frac{1}{1-\frac{1}{\log 4}}=:C_1,$ whence
$\int_0^\delta \log (1/u)du \leq C_1\delta\log \frac{1}{\delta}.$
We then have 
\begin{eqnarray*}
\int_{0}^{\hat R}\sqrt{ \log (4\|F\|_n/u)}du &=&
4\|F\|_n \int_{0}^{\hat R/(4\|F\|_n)}\sqrt{ \log (1/u)}du 
\\
&\leq & 4C_1\|F\|_n \hat R/(4\|F\|_n) \log \left( \frac{4\|F\|_n}{\hat R} \right)
\\
&=& C_1\hat R \log \left( \frac{4\|F\|_n}{\hat R} \right).
\end{eqnarray*}
We now show that $\frac{\|F\|_n}{\hat R} \leq \sqrt{n}.$ 
This follows since
$$\|F\|_n^2 = \frac{1}{n}\sum_{i=1}^n F(X_i)^2 \leq \max_{i=1,\dots,n} \sup_{f\in \mathcal F} f(X_i)^2,$$
and
$$\hat R^2 = \sup_{f\in \mathcal F} \frac{1}{n}\sum_{i=1}^n f(X_i)^2 \geq \sup_{f\in \mathcal F} \max_{i=1,\dots,n} f(X_i)^2/n.$$
Hence we obtain
$$\int_{0}^{\hat R}\sqrt{ \log (4\|F\|_n/u)}du \leq
 C_1\hat R \log (4\sqrt{n}) .$$
Then it follows
\begin{eqnarray*}
\mathbb E \sup_{f\in\mathcal F} |(\mathbb P_n-P)f| 
&\leq &
 2C_0\mathbb E {\hat R}\left( \sqrt{s\log p/n} + C_1\log (4\sqrt{n}) \sqrt{s/n}  \right)
\\
&\leq& C_3  \mathbb E {\hat R}\log (4\sqrt{n}) \sqrt{s\log p/n}
\end{eqnarray*}
We next apply the Dudley's inequality  to the class $\mathcal F^2:= \{f^2:f\in\mathcal F\}$.
First observe that
\begin{equation}\label{eb}
\log \mathcal N(\epsilon, \mathcal F^2,\|\cdot\|_n) \leq \log \mathcal N\left(\epsilon/\|F\|_{n,\infty}, \mathcal F,\|\cdot\|_n\right)
\leq
s\log p + s\log \left(\frac{\|F\|_n \|F\|_{n,\infty}}{\epsilon}\right).
\end{equation}
Let $\hat R_2^2 := \sup_{f\in\mathcal F^2} \|f\|_n^2 = \sup_{f\in\mathcal F^2} \frac{1}{n}\sum_{i=1}^n f(X_i)^2 = 
\sup_{f\in\mathcal F} \frac{1}{n}\sum_{i=1}^n f(X_i)^4.
$
Then since 
$$\mathbb E \sup_{f\in \mathcal F} |(\mathbb P_n-P)f^2| = \mathbb E \sup_{f\in \mathcal F^2} |(\mathbb P_n-P)f|,
$$
we obtain by Dudley's inequality and the entropy bound \eqref{eb} for $\mathcal F^2$ that
\begin{eqnarray*}
\mathbb E \sup_{f\in\mathcal F^2} |(\mathbb P_n-P)f| 
&\leq &
 \mathbb E\int_{0}^{\hat R_2}\sqrt{s\log p + s\log (4\|F\|_n \|F\|_{n,\infty}/u)}du /\sqrt{n}  
\\
&\lesssim &
 \mathbb E {\hat R_2}\sqrt{s\log p/n} + 
\underbrace{\mathbb E\int_{0}^{\hat R_2}\sqrt{ s\log (4\|F\|_n\|F\|_{n,\infty}/u)}du /\sqrt{n} }_{i} 
\end{eqnarray*}
We then have 
\begin{eqnarray*}
i &=&\int_{0}^{\hat R_2}\sqrt{ \log (\|F\|_n\|F\|_{n,\infty}/u)}du \\
&=&
4\|F\|_n \|F\|_{n,\infty} \int_{0}^{\hat R_2/(4\|F\|_n\|F\|_{n,\infty})}\sqrt{ \log (1/u)}du \\
&\leq &
 \|F\|_n\|F\|_{n,\infty} \hat R_2/\|F\|_n/\|F\|_{n,\infty} \log \left( \frac{4\|F\|_n\|F\|_{n,\infty}}{\hat R_2} \right)
\\
&=& \hat R_2 \log \left( \frac{4\|F\|_n\|F\|_{n,\infty}}{\hat R_2} \right).
\end{eqnarray*}
We now show that $\frac{\|F\|_n\|F\|_{n,\infty}}{\hat R_2} \leq \sqrt{n}.$ 
To this end observe that
$\|F\|_{n} \leq \|F\|_{n,\infty},$ and
$$\|F\|_{n,\infty}^2 
\leq \max_{i=1,\dots,n} \sup_{f\in \mathcal F} f(X_i)^2,$$
and moreover
$$\hat R^2_2 = \sup_{f\in \mathcal F} \frac{1}{n}\sum_{i=1}^n f(X_i)^4 \geq \sup_{f\in \mathcal F} \max_{i=1,\dots,n} f(X_i)^4/n.$$
Then it follows that
$$\frac{\|F\|_n\|F\|_{n,\infty}}{\hat R_2} \leq \frac{\|F\|_{n,\infty}^2}{\hat R_2} \leq \sqrt{n}.$$
Hence we obtain
$$i = \int_{0}^{\hat R_2}\sqrt{ \log (4\|F\|_n\|F\|_{n,\infty}/u)}du \leq
\hat R_2 \log (4\sqrt{n}).$$
Thus
\begin{eqnarray*}
\mathbb E \sup_{f\in\mathcal F} |(\mathbb P_n-P)f^2| 
&\lesssim &
 \mathbb E {\hat R_2} \log (4\sqrt{n}) \sqrt{s\log p/n} 
\end{eqnarray*}
By the triangle inequality, we obtain
$$\hat R^2 \leq \sup_{f\in\mathcal F} |(\mathbb P_n-P)f^2| + \sup_{f\in\mathcal F} Pf^2. $$
Then
$$\mathbb E\hat R \leq \mathbb E \sqrt{\sup_{f\in\mathcal F} |(\mathbb P_n-P)f^2| + \sup_{f\in\mathcal F} Pf^2} \leq
\mathbb E\sqrt{\sup_{f\in\mathcal F} |(\mathbb P_n-P)f^2|} +  \sqrt{\sup_{f\in\mathcal F} Pf^2}. $$
Using Jensen's inequality, we have 
$$\mathbb E (\sup_{f\in\mathcal F} |(\mathbb P_n-P)f^2|)^{1/2} \leq (
\mathbb E \sup_{f\in\mathcal F}
|(\mathbb P_n-P)f^2|)^{1/2}.$$
Hence
\begin{eqnarray*}
\mathbb E\sup_{f\in\mathcal F} |(\mathbb P_n-P)f| 
&\lesssim &
 \mathbb E {\hat R} \log n \sqrt{s\log p/n} 
\\
&\leq &
\left[(\mathbb E\sup_{f\in\mathcal F} |(\mathbb P_n-P)f^2|)^{1/2} + \sup_{f\in\mathcal F} {(\mathbb Ef^2)}^{1/2}\right]\log n \sqrt{s\log p/n} 
\\
&\leq& 
\left[ \mathbb E {\hat R_2}^{1/2} (\log n \sqrt{s\log p/n})^{1/2} + \sup_{f\in\mathcal F} \sqrt{\mathbb Ef^2}
\right]\log n \sqrt{s\log p/n} 
\\
&\leq &
 \mathbb E {\hat R_2}^{1/2} (\log n)^{3/2} [s\log p/n]^{1/4} \sqrt{s\log p} /\sqrt{n}\\[10px]
&&   + \sup_{f\in\mathcal F} \sqrt{\mathbb Ef^2} \log n \sqrt{s\log p/n} 
.
\end{eqnarray*}
Next observe that by Jensen's inequality
\begin{eqnarray*}
\mathbb E\sup_{f\in\mathcal F}\|f\|_{n,4} 
&=&
\mathbb E \left(\sup_{f\in\mathcal F}\s  f(X_i)^4\right)^{1/4}
\\
&\leq &
\left(\mathbb E \sup_{f\in\mathcal F}\s  f(X_i)^4\right)^{1/4}
\\
&=&\left(\s \mathbb E\sup_{f\in\mathcal F} f(X_i)^4\right)^{1/4}.
\\
&\leq &\left(\max_{i=1,\dots,n} \mathbb E\sup_{f\in\mathcal F} f(X_i)^4\right)^{1/4}.
\end{eqnarray*}
By assumption, we have $\max_{i=1,\dots,n}(\mathbb E\sup_{f\in\mathcal F}f(X_i)^4)^{1/4} (\log n\sqrt{s\log p/n})^{3/2} =o_P(1/\sqrt{n})$ 
 and 
$$\sqrt{\mathbb Ef^2}\log n \sqrt{ s\log p/n} =o(1/\sqrt{n}).$$ Hence
\begin{eqnarray*}
\mathbb E\sup_{f\in\mathcal F} |(\mathbb P_n-P)f| 
&=& o(1/\sqrt{n})
.
\end{eqnarray*}
Then
\begin{eqnarray*}
\mathbb E\sup_{f\in\mathcal F} |\mathbb G_nf| =o(1).
\end{eqnarray*}
By Markov's inequality it follows that $\sup_{f\in\mathcal F} |\mathbb G_nf| = o_P(1).$

\end{proof}

\subsection{Proofs for Section \ref{sec:glm.main} (High-dimensional generalized linear models) }
We first need the following preliminary lemmas before proving the statement of Theorem \ref{glm.main}.
\label{sec:glm.proof}
\subsubsection{Preliminary lemmas}

\begin{lemma}\label{glm.Pf2}
Consider the class 
$$\mathcal F:=\{\Theta_j^Tx(w_{\beta}-w_{\beta_0}): \|\beta-\beta_0\|_1 \leq s\lambda, 
\mathbb E(x^T(\beta-\beta_0))^2 \leq s\lambda^2
\}.$$
Assume  conditions \ref{itm:bounded}, \ref{itm:w}.
Then
$$\sup_{f\in\mathcal F}\|f\| =\mathcal O\left( s^{3/4}\lambda\right).$$
\end{lemma}

\begin{proof}[Proof of Lemma \ref{glm.Pf2}]
By the Lipschitz property in condition \ref{itm:w}, we have $|w_\beta-w_{\beta_0}| \leq L |x^T(\beta-\beta_0)|.$ By condition \ref{itm:bounded} and by H\"older's inequality, we obtain $|x(\beta-\beta_0)|\leq K_X \|\beta-\beta_0\|_1$. By condition \ref{itm:bounded}, we further have $\mathbb E|\Theta_j^Tx|^4 = O(1)$. Hence 
\begin{eqnarray*}
Pf^2 & = & \mathbb E\left(\Theta_j^T x (w_\beta-w_{\beta_0})\right)^2 
\\
&\leq & 
(\mathbb E|\Theta_j^Tx|^4 )^{1/2} (\mathbb E(w_\beta-w_{\beta_0})^{4})^{1/2} 
\\
&\leq  &
(\mathbb E|\Theta_j^Tx|^4 )^{1/2} (\mathbb E(w_\beta-w_{\beta_0})^{4})^{1/2} 
\\
&\leq  &
(\mathbb E|\Theta_j^Tx|^4 )^{1/2} K_X\|\beta-\beta_0\|_1(\mathbb E(w_\beta-w_{\beta_0})^{2})^{1/2} 
\\
& \leq & \mathcal O(1) s\lambda (\mathbb E (x^T (\beta-\beta_0))^2)^{1/2}
=\mathcal O\left( s^{3/2}\lambda^2 \right).
\end{eqnarray*}

\end{proof}

\begin{lemma}\label{f4}
Consider the class 
$$\mathcal F:=\{\Theta_j^Tx(w_{\beta}-w_{\beta_0}): \|\beta-\beta_0\|_1 \leq s\lambda, 
\mathbb E(x^T(\beta-\beta_0))^2 \leq s\lambda^2
\},$$
where $\lambda\asymp \sqrt{\log p/n}.$
Suppose that condition  \ref{itm:w} is satisfied, assume that $P|\Theta_j^Tx|^4 = O(1)$ and $s^{3}(\log p)^2(\log n)^2 / n = o(1).$
Then 
the condition 
$$P\sup_{f\in\mathcal F} f^4 s^3 (\log p)^3 (\log n)^6=o(1)$$
 is satisfied.
\end{lemma}
\begin{proof}[Proof of Lemma \ref{f4}]
By Lipschitz property of $w_\beta$ in condition \ref{itm:w}, we have $|w_\beta-w_{\beta_0}| \leq L |x(\beta-\beta_0)|,$ and hence
\begin{eqnarray*}
P\sup _{f\in\mathcal F}f^4 & \leq  &
 P|\Theta_j^T x|^4 \sup _{f\in\mathcal F} |w_\beta-w_{\beta_0}|^4
\\
&\leq &
 P|\Theta_j^T x|^4 L \|\beta-\beta_0\|_1^4
\\
&\leq &
\mathcal O(1) (s\lambda)^4.
\end{eqnarray*}
Then under $s^{3}(\log p)^2(\log n)^2 / n = o(1),$ the claim follows.
\end{proof}

\begin{lemma}\label{glm.cond1}
Assume 
conditions \ref{itm:bounded} and \ref{itm:ew}. 
Then the Condition \ref{taylor} is satisfied with  $d^2(\beta,\beta_0)= \mathbb E|x^T(\beta-\beta_0)|^2.$
\end{lemma}

\begin{proof}[Proof of Lemma \ref{glm.cond1}]
First note that 
$P\psi_\beta  = \int x \int w(y,x\beta)dP_{\epsilon|X} dP_X$ and $(P\psi_\beta)'_{\beta} =$ \linebreak $ \int x G'(u)|_{u=x\beta} dP_X.$
\begin{eqnarray*}
P(\psi_{\beta} (y,x) - \psi_{\beta_0}(y,x)) &=&
P (w(y,x^T\beta) - w(y,x^T\beta_0))x \\
&=&
\int (w(y,x^T\beta) - w(y,x^T\beta_0))x dP\\
&=& 
\int\int  (w (y,x^T\beta) - w(y,x^T\beta_0))x dP_{\epsilon|X} dP_X\\
&=& 
\int x\int  (w (y,x^T\beta) - w(y,x^T\beta_0)) dP_{\epsilon|X} dP_X\\
&=&
\int x G'(x^T\beta_0)x^T(\beta-\beta_0)dP_X \\
&&+\underbrace{
\int x (G'(\tilde z) -G'(x^T\beta_0))x^T(\beta-\beta_0) dP_X
}_{rem},
\end{eqnarray*}
where $\tilde z \in [x^T\beta_0,x^T\beta].$
Then for the remainder we have
\begin{eqnarray*}
\|\text{rem}\|_\infty &=&\| \int x (G'(\tilde z) -G'(x^T\beta_0) )x^T(\beta-\beta_0) dP_X \|_\infty
\\
& \leq & 
\int \|x\|_\infty |G'(\tilde z) -G'(x^T\beta_0)| |x^T(\beta-\beta_0)|  dP_X 
\\
& \leq & 
\int K_X L|x^T(\beta-\beta_0)|^2  dP_X 
\end{eqnarray*}
\end{proof}

\begin{lemma}\label{lipschitz}
Denote $h(\beta) := P\psi_\beta.$ Assume that $\beta\mapsto \frac{\partial ^2 (P\psi_\beta)_{j}}{\partial \beta_k\partial \beta_i}$ is bounded, i.e. $\frac{\partial ^2 h_j}{\partial \beta_k\partial \beta_i}$ is bounded for all $j=1,\dots,p$.
Then Condition \ref{itm:taylor} is satisfied with $d(\beta,\beta_0) = \|\beta-\beta_0\|_1.$
\end{lemma}
\begin{proof}[Proof of Lemma \ref{lipschitz}]
By the mean-value theorem,
\begin{eqnarray*}
(P\psi_\beta - P\psi_{\beta_0})_j = ((P\psi_\beta)_j)'|_{\beta=\beta_0} (\beta-\beta_0) + 
(((P\psi_\beta)_j)'|_{\beta=\tilde \beta} - ((P\psi_\beta)_j)'|_{\beta=\beta_0})(\beta-\beta_0) 
\end{eqnarray*}
Then for the remainder, we obtain
\begin{eqnarray*}
&&|((P\psi_\beta)_j)'|_{\beta=\tilde \beta} - ((P\psi_\beta)_j)'|_{\beta=\beta_0})(\beta-\beta_0) |\\
&\leq &
\| (P\psi_\beta)_j)'|_{\beta=\tilde \beta} - ((P\psi_\beta)_j)'|_{\beta=\beta_0} \|_\infty
\|\beta-\beta_0\|_1.
\end{eqnarray*}
Finally,
\begin{eqnarray*}
|((P\psi_\beta)_j)'|_{\beta=\tilde \beta} - ((P\psi_\beta)_j)'|_{\beta=\beta_0})_k|
& = & 
|((P\psi_\beta)_j)_k''|_{\beta=\bar \beta} (\beta-\beta_0) |
\\
&\ \leq & 
\|((P\psi_\beta)_j)_k''|_{\beta=\bar \beta} \|_\infty \|\beta-\beta_0\|_1
\\
&\ \leq & 
L \|\beta-\beta_0\|_1
\end{eqnarray*}

\end{proof}

\begin{proof}[Proof of Lemma \ref{glm.entropy.sparse}]
By the  Lipschitz property \ref{itm:w} it follows that
$$|w_\beta-w_{\beta_0}| \leq L |x^T(\beta-\beta_0)|, $$
hence
 $$N(\mathcal F, \|\cdot\|_{n},\epsilon) \leq 
N(\mathcal H, \|\cdot\|_{n},\epsilon/L)
,$$
where 
$\mathcal H=\{x\mapsto \Theta_j^T x x^T(\beta-\beta_0):\|\beta\|_0 \leq Cs\}$.
Let $V_i\subset \{1,\dots,p\}$ for $i=1,\dots,{p\choose s}$ be all subsets of $\{1,\dots,p\}$ of size $s$. 
We can rewrite
 $$\mathcal H = \bigcup_{i=1}^{{p\choose s}} \mathcal H_i ,$$
where $\mathcal H_i :=
\{(x, y)\mapsto \Theta_j^Tx x^T(\beta-\beta_0)):
 \beta\in\mathbb R^p, \beta_{V_i} = \mathbf 0 
 \}.$
The collection $\mathcal H_i$ has the same VC-index as an $s-$dimensional real vector space, which is $s+2$ by Lemma 2.6.15 in \cite{weak}.
\\
Then by Theorem 2.6.7 in \cite{weak} and since the covering number of a union of sets is upper bounded by sum of the covering numbers, we obtain
\begin{eqnarray*}
N(\epsilon\|H\|_n, \mathcal H, \|\cdot\|_n)
&\leq &
{p\choose s } N(\epsilon\|H\|_n, \mathcal H_i, \|\cdot\|_n)\\
& \leq & 
{p\choose s }  K V(\mathcal H_i) (16e)^{V(\mathcal H_i)} \left(\frac{1}{\epsilon}\right)^{2(V(\mathcal H_i)-1)}\\
&=&
{p\choose s }  Ks (16e)^{s} \left(\frac{1}{\epsilon}\right)^{2(s-1)},
\end{eqnarray*}
where $K$ is a universal constant and $0<\epsilon<1$. Then
\begin{eqnarray*}
{p\choose s }  Ks (16e)^{s} \left(\frac{1}{\epsilon}\right)^{2(s-1)}
&\leq &
\frac{p^s}{s!} Ks (16e)^{s} \left(\frac{1}{\epsilon}\right)^{2s}\\
&\leq &
\frac{p^s}{(s-1)!} K \left(\frac{16 e}{\epsilon}\right)^{2s}\\
&\leq&
K p^s  \left(\frac{16 e}{\epsilon}\right)^{2s}.
\end{eqnarray*}
Hence
$$N(\mathcal H , \|\cdot\|_{n},\epsilon\|H\|_n) 
 \lesssim \left(\frac{p}{\epsilon}\right)^s.
$$
Since $\|F\|_n \leq \|H\|_n$ and by the Lipschitz property of $w_\beta$ we have
\begin{eqnarray*} 
N(\mathcal F , \|\cdot\|_{n},\epsilon\|F\|_n) 
&\leq  &
N(\mathcal F , \|\cdot\|_{n},\epsilon\|H\|_n) \\
&\leq &
N(\mathcal H , \|\cdot\|_{n},\epsilon\|H\|_n/L) \\
&\lesssim &
\left(\frac{p}{\epsilon}\right)^s
.
\end{eqnarray*}

\end{proof}

\begin{proof}[Proof of Theorem \ref{glm.main}]
We apply Theorems \ref{mai} and \ref{main2}. 
By condition \ref{itm:w}, the function $u\mapsto \rho(u,y)$ is differentiable, and hence
the Karush-Kuhn-Tucker (KKT) conditions for the optimization problem defining $\hat\beta$ read
$$\frac{1}{n}\sum_{i=1}^n \dot \rho(y,x\hat\beta)x+\lambda \hat Z=0,$$
where $\hat Z$ is the sub-differential of the $\ell_1$ norm evaluated at $\hat\beta.$ Then taking $\psi_{\beta}(x,y):=\dot \rho(y,x\beta)x,$
it follows by the KKT conditions  that
$\|\frac{1}{n}\sum_{i=1}^n \psi_{\hat\beta}(x_i,y_i)\|_\infty = \|\lambda \hat Z\|_\infty = \mathcal O_P(\lambda).$
Hence the estimating equations are approximately satisfied.\\
Now we check conditions \ref{itm:eig0} - \ref{itm:norm}.\\
Condition \ref{eig0} follows by condition \ref{itm:bounded}.
\\
Condition \ref{itm:taylor}:
Under \ref{itm:w}, the condition of Lemma \ref{glm.cond1} is satisfied and thus the lemma yields that condition \ref{itm:taylor} of Theorem \ref{mai} is satisfied, with $d^2(\beta,\beta_0)  = \mathbb E |x^T(\beta-\beta_0)|^2.$ Then $d(\hat\beta,\beta_0) 
= 
(\hat\beta-\beta_0)^T \Sigma (\hat\beta-\beta_0) =\mathcal O(s\lambda^2) = o(1/\sqrt{ns}),
$
under the condition $s^3(\log p)^2 /n =o(1).$
\\
Condition \ref{itm:initial}: is satisfied under condition \ref{itm:initial.glm}.
\\
Condition \ref{itm:empp}:
We now show that the entropy condition of Theorem \ref{main2} is satisfied. Consider the class  of functions
$$\mathcal F=\{\Theta_j^T x(w_{\beta}-w_{\beta_0}): \|\beta-\beta_0\|_{1} \leq s\lambda, \|\beta\|_0 \leq s, \mathbb E|x^T(\beta-\beta_0)|^2 \leq s\lambda^2\},$$
where $\lambda\asymp \sqrt{\log p/n}$.
Under the condition \ref{itm:initial.glm}, it follows that $f_{\hat\beta}\in\mathcal F$ with high probability.\\
We proceed to check the entropy condition of  Theorem \ref{main2}.
Under condition \ref{itm:w}, Lemma \ref{glm.entropy.sparse} implies that the entropy bound is satisfied for the class $\mathcal F$.
Finally, we check condition \eqref{condi} of Theorem \ref{main2}. By Lemma \ref{glm.Pf2} it follows that
$R=\mathcal O(s^{3/2}\sqrt{\log p /n})$. Hence under $s^{5/2}(\log p)^2(\log n)^2 /n =o(1)$ it holds that $R\sqrt{s\log p}\log n=o(1).$
The condition 
$$P\sup_{f\in\mathcal F}f^4 s^3 (\log p)^3(\log n)^6/n=o(1)$$ is satisfied under the sparsity condition $s^{3}(\log p)^2(\log n)^2 /n =o(1)$  by Lemma \ref{f4}.
\\
Condition \ref{itm:norm}: by the assumption $\mathbb E(\Theta_j^T X_i)^4 = \mathcal O(1)$ and by \ref{itm:norm}, we can apply the central limit theorem to conclude the asymptotic normality as required in condition \ref{itm:norm}.
\\
The above implies that 
$$\sqrt{n}(\tilde b_j - \beta^0_j)/
\sqrt{\Theta_j^T P\psi_{\beta_0}\psi_{\beta_0}^T \Theta_j} \rightsquigarrow \mathcal N(0,1),$$
which shows the first claim of the theorem.\\
Next by Theorem \ref{my} it follows that $$\|\hat\Theta_j-\Theta_j\|_1 = \Op (s^{3/2}\sqrt{\log p/n}).$$
Furthermore, by Lemma \ref{var}, we have 
$$|\hat\Theta_j^T \mathbb P_n\psi_{\hat\beta}\psi_{\hat\beta}^T \hat\Theta_j - \Theta_j^T P\psi_{\beta_0}\psi_{\beta_0}^T \Theta_j| = o_P(1).$$
But then 
\begin{eqnarray*}
 (\hat\beta-\beta_0 -\hat\Theta_j^T \mathbb P_n \psi_{\hat\beta})/\sqrt{\hat\Theta_j^T \mathbb P_n\psi_{\hat\beta}\psi_{\hat\beta}^T \hat\Theta_j}
&=& 
(\hat\beta-\beta_0 -\Theta_j^T \mathbb P_n \psi_{\hat\beta})/\sqrt{\hat\Theta_j^T \mathbb P_n\psi_{\hat\beta}\psi_{\hat\beta}^T \hat\Theta_j}
\\
&&-\;(\hat\Theta_j-\Theta_j)^T \mathbb P_n \psi_{\hat\beta} /\sqrt{\hat\Theta_j^T \mathbb P_n\psi_{\hat\beta}\psi_{\hat\beta}^T \hat\Theta_j}.
\end{eqnarray*}

\end{proof}

\begin{proof}[Proof of Theorem \ref{glm.main2}]
We apply Theorems \ref{mai} and \ref{main2}. Under condition \ref{itm:ew}, by Lemma \ref{glm.cond1} it follows that Condition \ref{itm:taylor} is satisfied with $d^2(\beta,\beta_0)=\mathbb E|x^T(\beta-\beta_0)|^2$. Then $d^2(\hat\beta,\beta_0) 
= 
(\hat\beta-\beta_0)^T \Sigma (\hat\beta-\beta_0) =\mathcal O(s\lambda^2) = o(1/\sqrt{ns}),
$
under the condition $s^3(\log p)^2 /n =o(1).$\\
By inspection of proof of Lemma \ref{glm.Pf2}, we have
 by boundedness of $w_\beta$ (condition \ref{itm:ew}) that $|w_\beta-w_{\beta_0}| \leq |w_\beta|+|w_{\beta_0}| =\mathcal O(1)$ and hence
\begin{eqnarray*}
Pf^2 & = & \mathbb E\left(\Theta_j^T x (w_\beta-w_{\beta_0})\right)^2 
\\
&\leq & 
(\mathbb E|\Theta_j^Tx|^4 )^{1/2} (\mathbb E(w_\beta-w_{\beta_0})^{4})^{1/2} 
\\
&\leq  &
(\mathbb E|\Theta_j^Tx|^4 )^{1/2} (\mathbb E(w_\beta-w_{\beta_0})^{4})^{1/2} 
\\
& \leq & \mathcal O(1) (\mathbb E (w_\beta-w_{\beta_0})^2)^{1/2}
.
\end{eqnarray*}
 By condition \ref{itm:other} we have $\mathbb E \left(w_\beta-w_{\beta_0}\right)^2=\mathcal O\left(\frac{1}{s^2 (\log p)^2 (\log n)^4}\right)$  which implies that $R\log n \sqrt{s\log p}=o(1).$
\\
We have 
\begin{eqnarray*}
\max_{i=1,\dots,n}\mathbb E\sup _{f\in\mathcal F}f^4(X_i,Y_i) & \leq  &
\max_{i=1,\dots,n}\mathbb E |\Theta_j^T X_i|^4 \sup _{f\in\mathcal F} |w_\beta(X_i,Y_i)-w_{\beta_0}(X_i,Y_i)|^4
\\
&\leq &
 \max_{i=1,\dots,n}\mathbb E{|\Theta_j^T X_i|^4}
\\
&\leq &
\mathcal O(1) ,
\end{eqnarray*}
where we used the boundedness of $w_\beta$ from condition \ref{itm:ew} and the assumption 
$$\max_{i=1,\dots,n}\mathbb E|\Theta_j^T x|^4=\mathcal O(1)$$ from condition \ref{itm:bounded}.
The condition 
$$\max_{i=1,\dots,n}\mathbb E\sup_{f\in\mathcal F} f(X_i,Y_i)^4 s^3 (\log p)^3 (\log n)^6/n=o(1)$$ is then satisfied under 
$s^3 (\log p)^3 (\log n)^6/n=o(1).$
\end{proof}

\subsubsection{Proofs for Section \ref{subsec:spa}: Sparsity of the Lasso}

\begin{proof}[Proof of Lemma \ref{sparsity.glm}]
The KKT conditions for $\hat\beta$ give
\begin{equation*}
\mathbb P_n \psi_{\hat\beta} + \lambda\hat Z=0,
\end{equation*}
where $\psi_\beta(y_i,x_i)=w(y_i,x_i^T\beta)$.
This can be rewritten as
\begin{equation}\label{spa.beta}
\mathbb P_n (\psi_{\hat\beta}-\psi_{\beta_0}) =- \lambda\hat Z-\mathbb P_n \psi_{\beta_0}.
\end{equation}
Then we further separate the empirical process part
\begin{equation}\label{spa2}
P (\psi_{\hat\beta}-\psi_{\beta_0}) =- \lambda\hat Z-\mathbb P_n \psi_{\beta_0} - (\mathbb P_n-P) (\psi_{\hat\beta}-\psi_{\beta_0}).
\end{equation}
Taking the $\ell_2$-norm of the left-hand side of \eqref{spa.beta} and by the mean-value theorem we obtain 
\begin{eqnarray*}
\|P (\psi_{\hat\beta}-\psi_{\beta_0})\|_2^2
 &=&
\| P x_i(w(y_i,x_i^T\hat\beta) - w(y_i,x_i^T\beta_0))\|_2^2 
\\
&=&
\|P_xG'(y_i,x_i^T\tilde\beta) x_ix_i^T(\hat \beta - \beta_0) \|_2^2\\
&=&
(\hat\beta-\beta_0)^T(P_xG'(y_i,x_i^T\tilde\beta) x_ix_i^T)^2 (\hat \beta - \beta_0)
\end{eqnarray*}
For all $u\in\mathbb R^p$ we have (since $G'(y_i,x_i^T\tilde\beta) \geq 0$ by assumption of convexity of the loss function)
$$u^T P_xG'(y_i,x_i^T\tilde\beta) x_ix_i^T u \leq C u^T P_x x_ix_i^T u \leq C\Lambda_{\max}(P_x x_ix_i^T) u^Tu \leq C_2 u^Tu,$$
for some constant $C_2>0.$ Thus it must necessarily hold that $\Lambda_{\max}(P_xG'(y_i,x_i^T\tilde\beta) x_ix_i^T)\leq C_2.$
But then we obtain
\begin{eqnarray*}
\|P (\psi_{\hat\beta}-\psi_{\beta_0})\|_2^2
&\leq &
\Lambda_{\max}^2(P_xG'(y_i,x_i^T\tilde\beta) x_ix_i^T)
\|\hat \beta - \beta_0\|_2^2
\\
&=&C_2
\|\hat \beta - \beta_0\|_2^2
\\
&=&\op_P (\|W_{\beta_0} X(\hat\beta-\beta_0)\|_2^2),
\end{eqnarray*}
where $W_{\beta_0}:=\text{diag}(w_{\beta_0}(x_1,y_1),\dots,w_{\beta_0}(x_n,y_n)).$
Next we consider the right-hand side of \eqref{spa.beta}. 
First, by equation \eqref{ent.glm} which is assumed in the conditions, we have for the empirical process part
$$(\mathbb P_n-P) (\psi_{\hat\beta}-\psi_{\beta_0}) = \Op(\lambda_0),$$
for $\lambda_0 \asymp \sqrt{\log p/n}.$ This follows analogously as in the proof of Theorem \ref{glm.main2}.
We further have,
$$\mathbb P_n \psi_{\beta_0} = \Op(\lambda_0).$$
Denote $\hat s:=\|\hat\beta\|_0.$ Then taking the $\ell_2$-norm of the right-hand side of \eqref{spa.beta}
$$\|- \lambda\hat Z-\mathbb P_n \psi_{\beta_0} -(\mathbb P_n-P) (\psi_{\hat\beta}-\psi_{\beta_0}) \|_2^2 
 \geq (\lambda - \lambda_0)^2\hat s.$$
Hence we obtain
$$\hat s \leq \frac{\op_P (\|W_{\beta_0} X(\hat\beta-\beta_0)\|_2^2) }{\lambda^2} \leq \op_P (s).$$
\end{proof}

\subsubsection{Proofs for Section \ref{subsec:var}: Estimation of asymptotic variance}

\begin{proof}[Proof of Lemma \ref{var}]
\begin{eqnarray*}
&&
|\hat\Theta_j^T \mathbb P_n\psi_{\hat\beta}\psi_{\hat\beta}^T \hat\Theta_j - \Theta_j^T P\psi_{\beta_0}\psi_{\beta_0}^T \Theta_j|
\\
&
\leq
&
\underbrace{
|\hat\Theta_j^T (\mathbb P_n-P)\psi_{\beta_0}\psi_{\beta_0}^T \hat\Theta_j|
}_{i}
+
\underbrace{
|\hat\Theta_j^T P\psi_{\beta_0}\psi_{\beta_0}^T \hat\Theta_j  - \hat\Theta_j^T P\psi_{\hat\beta}\psi_{\hat\beta}^T \hat\Theta_j |
}_{ii }
\end{eqnarray*}
For the first term, we have by H\"older's inequality
$$|i| \leq \|\hat\Theta_j\|_1^2 \|(P_n-P)\psi_{\beta_0}\psi_{\beta_0}^T\|_\infty = \Op (s \sqrt{\log p/n}).$$
For the second term, we have
\begin{eqnarray*}
|P(\psi_{\beta_0}\psi_{\beta_0}^T-\psi_{\hat\beta}\psi_{\hat\beta}^T)| 
&=&
|P(w_{\hat\beta}^2 - w_{\beta_0}^2)xx^T|
\\
& \leq &
\|x\|_\infty^2 |P(w_{\hat\beta}^2-w_{\beta_0}^2)|
\\
&=& \mathcal O (\mathbb E x^T(\hat\beta-\beta_0))\\
&=& \Op ( (\mathbb E (x^T(\hat\beta-\beta_0))^2)^{1/2} )=\Op (\sqrt{s}\lambda).
\end{eqnarray*}
Then
$$|ii|\leq \|\hat\Theta_j\|_1^2\|x\|_\infty^2  |P(w_{\hat\beta}^2-w_{\beta_0}^2)| = \Op (s \sqrt{s}\lambda) = \Op (s^{3/2}\lambda)=o_P(1).$$
\end{proof}

\subsection{Proofs for section \ref{sec:nw} (Nodewise regression for estimation of precision matrices) }
\label{sec:nw.proof}

\begin{lemma}\label{node.rates}
Suppose that conditions \ref{itm:bounded}, \ref{itm:initial.glm}, \ref{eig}, \ref{A3}, \ref{sp} are satisfied.
Let $\lambda \asymp \sqrt{\log p/n} $.
Then it holds that
$$\|\hat\gamma_{\hat\beta,j} -  \gamma_{0,j}\|_1 = \Op (s^{3/2} \sqrt{\log p/n}).$$
\end{lemma}

\begin{proof}[Proof of Lemma \ref{node.rates}]
We denote $W_\beta:=\text{diag}({v}_{\beta}(y_1,x_1),\dots,{v}_{\beta}(y_n,x_n) )$.
Further denote by 
$$\eta_{\beta_0,j}:= W_{\beta_0}(X_j - X_{-j}\gamma_{0,j}).$$
We have the basic inequality
\begin{eqnarray*}
&&\|X_{\hat\beta,-j}(\hat\gamma_{\hat\beta,j}-\gamma_{\beta_0,j})\|_2^2/n + 2\lambda_j \|\hat\gamma_{\hat\beta,j}\|_1
\\
&&\leq 2 \eta_{\beta_0,j}^T W_{\beta_0}^{-1}W_{\hat\beta}X_{\hat\beta,-j}(\hat\gamma_{\hat\beta,j}-\gamma_{\beta_0,j})/n
+2\lambda_j \|\gamma_{\beta_0,j}\|_1.
\end{eqnarray*}
First we have by the Cauchy-Schwarz inequality
\begin{eqnarray*}
&&| \eta_{\beta_0,j}^T W_{\beta_0}^{-1}W_{\hat\beta}X_{\hat\beta,-j}(\hat\gamma_{\hat\beta,j}-\gamma_{\beta_0,j})/n - 
\eta_{\beta_0,j}^T X_{\beta_0,-j}(\hat\gamma_{\hat\beta,j}-\gamma_{\beta_0,j})/n
|\\
& \leq &
\|(W_{\hat\beta}^2W_{\beta_0}^{-2}-I)\eta_{\beta_0,j}\|_2/\sqrt{n} \|X_{\beta_0,-j}(\hat\gamma_{\hat\beta,j}-\gamma_{\beta_0,j})\|_2/\sqrt{n}
\end{eqnarray*}
We bound the term $\|(W_{\hat\beta}^2W_{\beta_0}^{-2}-I)\eta_{\beta_0,j}\|_2^2/n.$
We have by the Cauchy-Schwarz inequality
\begin{eqnarray*}
\|(W_{\hat\beta}^2W_{\beta_0}^{-2}-I)\eta_{\beta_0,j}\|_2^2/n  
&=&
\frac{1}{n}\sum_{i=1}^n ({v}_{\hat\beta,i}^2{v}_{\beta_0,i}^{-2}-1)^2 \eta_{\beta_0,j,i}^2
\\
&\leq & \sqrt{\frac{1}{n}\sum_{i=1}^n ({v}_{\hat\beta,i}^2{v}_{\beta_0,i}^{-2}-1)^4 }\sqrt{\frac{1}{n}\sum_{i=1}^n\eta_{\beta_0,j,i}^4}
\end{eqnarray*}
 By condition \ref{itm:bounded} and \ref{eig} we have $\mathbb E\eta_{\beta_0,j,i}^4 = \mathcal O(1).$
Then, and by the law of large numbers,
we have
$$\frac{1}{n}\sum_{i=1}^n\eta_{\beta_0,j,i}^4=\mathcal O_{P}(1).$$
For the other term we have since ${v}_{\beta_0,i}^{-2}=\mathcal O(1)$ by condition \ref{A3} that
$$\frac{1}{n}\sum_{i=1}^n ({v}_{\hat\beta,i}^2{v}_{\beta_0,i}^{-2}-1)^4
=\frac{1}{n} \sum_{i=1}^n {v}_{\beta_0,i}^{-2}({v}_{\hat\beta,i}^2-{v}_{\beta_0,i}^{2})^4
=\mathcal O\left(\frac{1}{n} \sum_{i=1}^n ({v}_{\hat\beta,i}^2-{v}_{\beta_0,i}^{2})^4\right).
$$
On this term we apply the Lipschitz property of ${v}_{\beta,i}$ to have 
\begin{eqnarray*}
|{v}_{\hat\beta,i}^2-{v}_{\beta_0,i}^2| &=& |({v}_{\hat\beta,i}-{v}_{\beta_0,i})^2 
+ 2{v}_{\beta_0,i}({v}_{\hat\beta,i}-{v}_{\beta_0,i})|
\\
&\leq & L |X_i^T(\hat\beta - \beta_0)| (L |X_i^T(\hat\beta - \beta_0)|+2).
\end{eqnarray*}
Now observe that 
$$|L |X_i^T(\hat\beta - \beta_0)|+2|\leq 2 + \|X_i\|_\infty\|\hat\beta - \beta_0\|_1=\op_P(1).$$
Then we obtain
\begin{eqnarray*}
\frac{1}{n} \sum_{i=1}^n ({v}_{\hat\beta,i}^2-{v}_{\beta_0,i}^{2})^4 
&=& 
\op_P(1)\frac{1}{n} \sum_{i=1}^n (L |X_i^T(\hat\beta - \beta_0)|)^4
\\
&\leq & 
\op_P(1)\max_{i=1,\dots,n}\|X_i\|_\infty^2 \|\hat\beta-\beta_0\|_1^2  
\frac{1}{n} \sum_{i=1}^n (L |X_i^T(\hat\beta - \beta_0)|)^2
,
\end{eqnarray*}
where in the last step we applied H\"older's inequality to one part of the term. This then gives by the result 
$\|X(\hat\beta-\beta_0)\|_2^2/n = \mathcal O_P(s\lambda^2)$, by $\max_{i=1,\dots,n}\|X_i\|_\infty^2=\mathcal O(K^2)$ and by the $\ell_1$ rates 
$\|\hat\beta-\beta_0\|_1=\mathcal O_P(s\lambda)$
that
\begin{eqnarray*}
\frac{1}{n} \sum_{i=1}^n ({v}_{\hat\beta,i}^2-{v}_{\beta_0,i}^{2})^4 
&\leq & 
\op_P(1) \max_{i=1,\dots,n}\|X_i\|_\infty^2 \|\hat\beta-\beta_0\|_1^2  
\frac{1}{n} \sum_{i=1}^n (L |X_i^T(\hat\beta - \beta_0)|)^2
\\
&=& \mathcal O_P(s^3\lambda^4).
\end{eqnarray*}
This is turn implies
$$\|(W_{\hat\beta}^2W_{\beta_0}^{-2}-I)\eta_{\beta_0,j}\|_2^2/n =\mathcal O_P(s^{3/2}\lambda^2).$$
Hence, and returning to the basic inequality, we obtain
\begin{eqnarray*}
\|X_{\hat\beta,-j}(\hat\gamma_{\hat\beta,j}-\gamma_{\beta_0,j})\|_2^2/n + 2\lambda_j \|\hat\gamma_{\hat\beta,j}\|_1
&\leq & 2\eta_{\beta_0,j}^T X_{\beta_0,-j}(\hat\gamma_{\hat\beta,j}-\gamma_{\beta_0,j})/n
\\
&&+
\mathcal O_P(s^{3/4}\lambda) \|X_{\beta_0,-j}(\hat\gamma_{\hat\beta,j}-\gamma_{\beta_0,j})\|_2/\sqrt{n}
\\
&&+2\lambda_j \|\gamma_{\beta_0,j}\|_1.
\end{eqnarray*}
For arbitrary $\delta>0$, we have
$$2ab\leq \delta a^2 + b^2/\delta.$$
Applying this claim we get 
$$2\mathcal O_P(s^{3/4}\lambda) \|X_{\beta_0,-j}(\hat\gamma_{\hat\beta,j}-\gamma_{\beta_0,j})\|_2/n \leq 
\mathcal O_P(s^{3/2}\lambda^2)/\delta + \delta\|X_{\beta_0,-j}(\hat\gamma_{\hat\beta,j}-\gamma_{\beta_0,j})\|_2^2/n.
$$
Now we use that, under conditions on ${v}_\beta$ in condition \ref{A3}, 
$$\|X_{\beta_0,-j}(\hat\gamma_{\hat\beta,j}-\gamma_{\beta_0,j})\|_2^2/n=\mathcal O_P(\|X_{\hat\beta,-j}(\hat\gamma_{\hat\beta,j}-\gamma_{\beta_0,j})\|_2^2/n).$$
Hence
 we get
\begin{eqnarray*}
(1-\delta)\|X_{\hat\beta,-j}(\hat\gamma_{\hat\beta,j}-\gamma_{\beta_0,j})\|_2^2/n + 2\lambda_j \|\hat\gamma_{\hat\beta,j}\|_1
&\leq & 2\eta_{\beta_0,j}^T X_{\beta_0,-j}(\hat\gamma_{\hat\beta,j}-\gamma_{\beta_0,j})/n
\\
&&+
\mathcal O_P(s^{3/2}\lambda^2)
+2\lambda_j \|\gamma_{\beta_0,j}\|_1.
\end{eqnarray*}
Now we have by H\"older's inequality
$$\eta_{\beta_0,j}^T X_{\beta_0,-j}(\hat\gamma_{\hat\beta,j}-\gamma_{\beta_0,j})/n 
\leq \|\eta_{\beta_0,j}^T X_{\beta_0,-j}\|_\infty/n \|\hat\gamma_{\hat\beta,j}-\gamma_{\beta_0,j}\|_1.$$
Under \ref{itm:bounded}, by Nemirovski's inequality (Theorem \ref{nemirovski}) it follows
$$\|\eta_{\beta_0,j}^T X_{\beta_0,-j}\|_\infty/n \leq \lambda_j,$$
for $\lambda_j\asymp \sqrt{\log p/n}$.
Hence
\begin{eqnarray*}
(1-\delta)\|X_{\hat\beta,-j}(\hat\gamma_{\hat\beta,j}-\gamma_{\beta_0,j})\|_2^2/n + 2\lambda_j \|\hat\gamma_{\hat\beta,j}\|_1
&\leq & \lambda_j \|\hat\gamma_{\hat\beta,j}-\gamma_{\beta_0,j}\|_1
\\
&&+\;
\mathcal O_P(s^{3/2}\lambda^2)
+2\lambda_j \|\gamma_{\beta_0,j}\|_1.
\end{eqnarray*}
Using triangle inequality we can get from the above
\begin{eqnarray*}
(1-\delta)\|X_{\hat\beta,-j}(\hat\gamma_{\hat\beta,j}-\gamma_{\beta_0,j})\|_2^2/n + \lambda_j \|\hat\gamma_{\hat\beta,j,S^c}\|_1 
\leq  3\lambda_j \|\hat\gamma_{\hat\beta,j,S}-\gamma_{\beta_0,j,S}\|_1 
+
\mathcal O_P(s^{3/2}\lambda^2).
\end{eqnarray*}
\textbf{Case i)}\\
If $\lambda \|\hat\gamma_{\hat\beta,j,S}-\gamma_{\beta_0,j,S}\|_1 \geq \mathcal O_P(s^{3/2}\lambda^2)$ then
\begin{eqnarray*}
(1-\delta)\|X_{\hat\beta,-j}(\hat\gamma_{\hat\beta,j}-\gamma_{\beta_0,j})\|_2^2/n + \lambda_j \|\hat\gamma_{\hat\beta,j,S^c}\|_1 
&\leq & 4\lambda_j \|\hat\gamma_{\hat\beta,j,S}-\gamma_{\beta_0,j,S}\|_1 
\end{eqnarray*}
Then we continue the chain of calculations
\begin{eqnarray*}
(1-\delta)\|X_{\hat\beta,-j}(\hat\gamma_{\hat\beta,j}-\gamma_{\beta_0,j})\|_2^2/n + \lambda_j \|\hat\gamma_{\hat\beta,j,S^c}\|_1 
&\leq & 4\lambda_j \|\hat\gamma_{\hat\beta,j,S}-\gamma_{\beta_0,j,S}\|_1 
\\
&\leq & 
4\lambda_j \sqrt{s} \|\hat\gamma_{\hat\beta,j,S}-\gamma_{\beta_0,j,S}\|_2
\\
&\leq &
4\lambda_j \sqrt{s} \|X_{\hat\beta,-j}(\hat\gamma_{\hat\beta,j}-\gamma_{\beta_0,j})\|_2/\sqrt{n}
\\
&\leq &
16s\lambda^2_j/\delta \\
&&+\; \delta\|X_{\hat\beta,-j}(\hat\gamma_{\hat\beta,j}-\gamma_{\beta_0,j})\|_2^2/n.
\end{eqnarray*}
That implies
\begin{eqnarray*}
(1-2\delta)\|X_{\hat\beta,-j}(\hat\gamma_{\hat\beta,j}-\gamma_{\beta_0,j})\|_2^2/n + \lambda_j \|\hat\gamma_{\hat\beta,j,S^c}\|_1 
\leq
16s\lambda^2_j/\delta.
\end{eqnarray*}
But then
$$\|\hat\gamma_{\hat\beta,j}-\gamma_{\beta_0,j}\|_1  = \|\hat\gamma_{\hat\beta,j,S}-\gamma_{\beta_0,j,S}\|_1  +
\|\hat\gamma_{\hat\beta,j,S^c}\|_1 =\mathcal O_P(s\lambda^2). $$
\textbf{Case ii)}\\
If $\lambda \|\hat\gamma_{\hat\beta,j,S}-\gamma_{\beta_0,j,S}\|_1 \leq \mathcal O_P(s^{3/2}\lambda^2)$ then
\begin{eqnarray*}
(1-\delta)\|X_{\hat\beta,-j}(\hat\gamma_{\hat\beta,j}-\gamma_{\beta_0,j})\|_2^2/n + \lambda_j \|\hat\gamma_{\hat\beta,j,S^c}\|_1 
&\leq & \mathcal O_P(s^{3/2}\lambda^2).
\end{eqnarray*}
But then
$$\|\hat\gamma_{\hat\beta,j}-\gamma_{\beta_0,j}\|_1  = \|\hat\gamma_{\hat\beta,j,S}-\gamma_{\beta_0,j,S}\|_1  +
\|\hat\gamma_{\hat\beta,j,S^c}\|_1 =\mathcal O_P(s^{3/2}\lambda^2). $$
\end{proof}

\begin{lemma}
\label{l2}
Suppose that conditions \ref{itm:bounded}, \ref{itm:initial.glm}, \ref{eig}, \ref{A3}, \ref{sp} are satisfied.
Let $\lambda \asymp \sqrt{\log p/n} $.
Then it holds that
$$|\hat\tau_j^2 - \tau_j^2| = O_P(K_X\sqrt{s\log p/n}).$$
\end{lemma}

\begin{proof}
By the definition of $\hat\tau_j^2$
$$\hat\tau_j^2 = X_{\hat\beta,j}^T( X_{\hat\beta,j} - X_{\hat\beta,-j}\hat\gamma_{\hat\beta,j}) /n=
X_j^T W_{\hat\beta}^2  ( X_{j} - X_{-j}\hat\gamma_{\hat\beta,j}) /n
.$$
We have
\begin{eqnarray*}
\hat\tau_j^2 - \tau_j^2 &=&
\underbrace{
X_j^T (W_{\hat\beta}^2 -W_{\beta_0}^2) ( X_{j} - X_{-j}\hat\gamma_{\hat\beta,j}) /n
}_{I}
+  
\underbrace{
X_j^T W_{\beta_0}^2  ( X_{j} - X_{-j}\hat\gamma_{\hat\beta,j}) /n - \tau_j^2}_{II}
.
\end{eqnarray*}
We treat the two terms separately.  We have 
\begin{eqnarray*}
|II| &=& 
X_j^T W_{\beta_0}^2  ( X_{j} - X_{-j}\hat\gamma_{\hat\beta,j}) /n - \tau_j^2 
\\
&=&
\underbrace{\|W_{\beta_0}(X_j-X_{-j}\gamma_{0,j})\|_2^2/n -\tau_j^2}_{=\op_P (1/\sqrt{n})} 
\\
&&+ 
\underbrace{\gamma_{0,j}^TX_{-j}^T W_{\beta_0}^2 (X_j - X_{-j} \gamma_{0,j})/n}_{=\op_P(\sqrt{s}\lambda)} 
\\
&&+
\underbrace{ X_j^T W_{\beta_0}^2X_{-j}(\hat\gamma_{\hat\beta,j} - \gamma_{0,j})/n }_{=\op_P(\lambda s^{3/2}\lambda)}
\end{eqnarray*}
Hence, 
$II= \op_P(\sqrt{s\log p/n}).$
For the first term, we have
\begin{eqnarray*}
|I| &=&
|X_j^T (W_{\hat\beta}^2 -W_{\beta_0}^2) ( X_{j} - X_{-j}\hat\gamma_{\hat\beta,j}) /n|
\\
&= &
\frac{1}{n}\sum_{i=1}^n X_{j,i} ({v}_{\hat\beta,i}^2 - {v}_{\beta_0,i}^2) \hat\eta_{j,i},
\end{eqnarray*}
where $\hat\eta_j := X_{j} - X_{-j}\hat\gamma_{\hat\beta,j}$.
Then
\begin{eqnarray*}
|I| &=&
\frac{1}{n}\sum_{i=1}^n X_{j,i} ({v}_{\hat\beta,i}^2 - {v}_{\beta_0,i}^2) \hat\eta_{j,i}
\\
&\leq &
\sqrt{ \frac{1}{n}\sum_{i=1}^n ({v}_{\hat\beta,i}^2 - {v}_{\beta_0,i}^2)^2 } \sqrt{\frac{1}{n}\sum_{i=1}^n X_{j,i}^2 \hat\eta_{j,i}^2}
\\
&\leq &
\sqrt{ \frac{1}{n}\sum_{i=1}^n ({v}_{\hat\beta,i}^2 - {v}_{\beta_0,i}^2)^2 } K_X \sqrt{\frac{1}{n}\sum_{i=1}^n  \hat\eta_{j,i}^2}.
\end{eqnarray*}
Then we have
\begin{eqnarray*}
 \frac{1}{n}\sum_{i=1}^n ({v}_{\hat\beta,i}^2 - {v}_{\beta_0,i}^2)^2 
&=& 
 \frac{1}{n}\sum_{i=1}^n (({v}_{\hat\beta,i} - {v}_{\beta_0,i})^2 + 2{v}_{\beta_0}({v}_{\hat\beta,i} - {v}_{\beta_0,i}) )^2 
\\
& \lesssim &
\frac{1}{n}\sum_{i=1}^n ({v}_{\hat\beta,i} - {v}_{\beta_0,i})^2
\\
& \leq &
\frac{1}{n}\sum_{i=1}^n (x_i^T(\hat\beta-\beta_0))^2
\\
&\lesssim &
s\log p/n.
\end{eqnarray*}
For the second term we have 
\begin{eqnarray*}
\hat\eta_j^T\hat\eta_j/n &=&
\|X_j- X_{-j}\hat\gamma_j\|_2^2/n
\\
&=&
\|X_j- X_{-j}\gamma^0_j\|_2^2/n
+
2(X_j- X_{-j}\gamma_j^0)^T  X_{-j}(\hat\gamma - \gamma^0_j)/n \\
&&+ \; 
\|X_{-j}(\hat\gamma - \gamma_j^0)\|_2^2/n.
\end{eqnarray*}
Let $\eta_j:=X_j-X_{-j}\gamma_j^0$. Then we have $\eta_j^T \eta_j/n = \mathcal O_P(\mathbb E\eta_j^T\eta_j)=\mathcal O_P(\tau_j^2)$.
Further we have
\begin{eqnarray*}
|(X_j- X_{-j}\gamma_0)^T  X_{-j}(\hat\gamma - \gamma_0)/n| 
&\leq &
 \|\eta_j^T\eta_j\|_2/\sqrt{n} \|X_{-j}(\hat\gamma - \gamma_0)\|_2 /\sqrt{n} \\
&=&\op_P(\tau_j^2) o_P({s}^{3/4}\lambda) = o_P(1).
\end{eqnarray*}
Hence 
$$\hat\eta_j^T\hat\eta_j/n =\mathcal O_P(\tau_j^2) = \mathcal O_P(1).$$
Therefore,
\begin{eqnarray*}
|I| 
&\leq &
\sqrt{ \frac{1}{n}\sum_{i=1}^n ({v}_{\hat\beta,i}^2 - {v}_{\beta_0,i}^2)^2 } K_X \sqrt{\frac{1}{n}\sum_{i=1}^n  \hat\eta_{j,i}^2}.
\\
& \lesssim &
K_X \sqrt{s\log p /n}. 
\end{eqnarray*}

\end{proof}


\begin{proof}[Proof of Theorem \ref{my}]
We use Lemmas \ref{node.rates} and \ref{l2} to obtain
\begin{eqnarray*}
\|\hat\Theta_j-\Theta_{j}^0\|_1 &=&
\|\hat\Gamma_j/\hat\tau_j^2 - \Gamma_j/\tau_j^2\|_1
\leq 
\underbrace{
\|\hat\gamma_j - \gamma^0_j\|_1 /\hat\tau_j^2}_{i} + 
\underbrace{
\|\gamma^0_j\|_1 (1/\hat\tau_j^2 - 1/\tau_j^2)}_{ii}.
\end{eqnarray*}
We have $i=\op_P(s^{3/2}\sqrt{\log p/n})$
and
$ii=\op_P(s\sqrt{\log p/n})$. 
\end{proof}

\begin{lemma}\label{nd.fi}
Suppose the generalized linear model setting from Example \ref{ex:glm}:  $Y_i=g(X_i^T\beta_0)+\epsilon_i$, where $X_i$ and $\epsilon_i$ are independent for $i=1,\dots,n$. Let 
 $\psi_{\beta}(x,y):=w(y-g(x^T\beta))x.$ Then
$$(\mathbb E\psi_\beta (x,y))'_{\beta}  = \mathbb E_x(\mathbb E_{Y}w(u,y))'_{u=x^T\beta_0}  x  x^T.$$
\end{lemma}
\begin{proof}[Proof of Lemma \ref{nd.fi}]
We have
\begin{eqnarray*}
\mathbb E\psi_\beta(y,x)
&=&\mathbb E_{(x,Y)} w(x^T\beta,y) x 
\\
&=& \mathbb E_x \mathbb E_{Y} (w(x^T\beta,y)|x) x.
\end{eqnarray*}
Let $G(u) := \mathbb E_{Y} w(u,y)$.
We use a Taylor expansion of $G$ around $u_0=x^T\beta_0$
$$G(u) = G(u_0) + G'(u_0)(u-u_0) + \frac{1}{2}G''(u)|_{u=\tilde u}(u-u_0)^2,$$
where $\tilde u =\alpha u_0 + (1-\alpha)u$ for some $0\leq \alpha \leq 1.$
This yields
\begin{eqnarray*}
\mathbb E_{Y} (w(x^T\beta,y)|x)  &=&
\mathbb E_{Y} (w(x^T\beta_0)|x) + G'(u)|_{u=x^T\beta_0} x^T(\beta -\beta_0)\\
&& + \frac{1}{2}G''(u)|_{u=\tilde u}(x^T(\beta -\beta_0))^2
.
\end{eqnarray*}
Differentiating this with respect to $\beta$ we obtain
\begin{eqnarray*}
(\mathbb E_{Y} (w(x^T\beta, y)|x))'_{\beta}  
&=&
-G'(u)|_{u=x^T\beta_0} x \\
&&+ \left(\frac{1}{2}G''(u)|_{u=\tilde u}\right)'_{\beta}(x^T(\beta -\beta_0))^2
\\
&&+ 
\frac{1}{2}G''(u)|_{u=\tilde u}2(x^T(\beta -\beta_0))x^T
.
\end{eqnarray*}
Plugging in $\beta=\beta_0$ we obtain
$$(\mathbb E_{Y} (w(x^T\beta, y)|x))'_{\beta=\beta_0}    =
-G'(u)|_{u=x^T\beta_0}x
.$$
Now we have 
$$(\mathbb E\psi_\beta(x,y))'_{\beta=\beta_0} = \mathbb E_x (\mathbb E_{Y} (w(x^T\beta, y)|x))'_{\beta=\beta_0} x.$$
Hence
$$(\mathbb E\psi_\beta(x,y))'_{\beta} = -\mathbb E_x G'(u)|_{u=x^T\beta_0} x x^T 
= -\mathbb E_x (\mathbb E_{Y}w(y,u)|x)'_{u=x^T\beta_0}  x  x^T.$$
\end{proof}

\subsection{Proofs for Section \ref{sec:examples} (Examples)}
\label{sec:examples.proof}
\subsubsection{Proofs for Section \ref{subsec:qre} ($\ell_1$-penalized LAD estimator)}
\begin{proof}[Proof of Lemma \ref{qr.ee}]
The necessary conditions for the problem \eqref{qrp} read
$$\mathbb P_n s_{\hat\beta} + \lambda \hat Z=0,$$
where $s_{\beta}$ is the subdifferential of $\beta\mapsto |y-x^T\beta|$ and $\hat Z$ is the subdifferential of the $\ell_1$ norm.
One can show that if $Y_i$ is absolutely continuous conditional on $X_i$, then with  probability one, 
$\|\{i:y_i = x_i^T\hat\beta\}\|_0= \|\hat \beta\|_0$ (i.e. exact interpolation happens exactly $\|\hat \beta\|_0$ times). 
The estimator $\hat\beta$ satisfies the condition of Lemma \ref{sparsity.glm}, therefore, $\|\hat\beta\|_0=\mathcal O_P(s).$
Hence,
\begin{eqnarray*}
\| \mathbb P_n \psi_{\hat\beta} \|_\infty & =& 
\| \mathbb P_n \psi_{\hat\beta} - \mathbb P_n s_{\hat\beta}\|_\infty  + 
\|\mathbb P_n s_{\hat\beta} \|_\infty
\\
& \leq & 
\|\frac{1}{n}\sum_{i: y_i\not = x_i^T\hat\beta} (\psi_{\hat\beta} - s_{\hat\beta})\|_\infty  + 
\|\frac{1}{n}\sum_{i: y_i = x_i^T\hat\beta} (\psi_{\hat\beta} - s_{\hat\beta})\|_\infty  + 
\mathcal O_P(\lambda)
\\
& \leq & 
 s/n +  \mathcal O_P(\lambda) = \mathcal O_P(\lambda),
\end{eqnarray*} 
where we used that $\|s_{\hat\beta}\|_\infty \leq \|\psi_{\hat\beta}\|_\infty =\mathcal O_P(1)$.
\end{proof}

 \begin{proof}[Proof of Lemma \ref{qr.entropy}]
Let $V_i\subset \{1,\dots,p\}$ for $i=1,\dots,{p\choose s}$ be all subsets of $\{1,\dots,p\}$ of size $s$. 
We can rewrite
 $$\mathcal F = \bigcup_{i=1}^{{p\choose s}} \mathcal F_i ,$$
where $\mathcal F_i :=
\{(x, y)\mapsto \Theta_j^Tx(1_{y \leq x^T \beta} - 1_{y \leq x^T \beta_0 }):
 \beta\in\mathbb R^p, \beta_{V_i} = \mathbf 0 
 \}.$
We now show that the class $\mathcal F_i$ has VC-index $V(\mathcal F_i) $ of order $s$. 
\\ 
The VC-index of a class of functions $\mathcal F_i$ is defined as the VC-index of the collection of  sets $\{(x,y,t): t < f(x,y)\}, $ where $t\in\mathbb R, f\in\mathcal F_i.$
The collection $\{(x,y)\mapsto y-x^T\beta: \beta_{V_i} =0\}$ has the same VC-index as an $s-$dimensional real vector space, which is $s+2$ by Lemma 2.6.15 in \cite{weak}.
Hence it follows that the VC-index of the collection $\{(x,y,t): t < f(x,y)\} $ is of order $s$.
\\
Then by Theorem 2.6.7 in \cite{weak} and since the covering number of a union of sets is upper bounded by sum of the covering numbers, we obtain
\begin{eqnarray*}
N(\epsilon\|F\|_n, \mathcal F, \|\cdot\|_n)
&\leq &
{p\choose s } N(\epsilon\|F\|_n, \mathcal F_i, \|\cdot\|_n)
\\
& \leq & 
{p\choose s }  K V(\mathcal F_i) (16e)^{V(\mathcal F_i)} \left(\frac{1}{\epsilon}\right)^{2(V(\mathcal F_i)-1)}\\
&=&
{p\choose s }  Ks (16e)^{s} \left(\frac{1}{\epsilon}\right)^{2(s-1)},
\end{eqnarray*}
where $K$ is a universal constant and $0<\epsilon<1$.
\begin{eqnarray*}
{p\choose s }  Ks (16e)^{s} \left(\frac{1}{\epsilon}\right)^{2(s-1)}
&\leq &
\frac{p^s}{s!} Ks (16e)^{s} \left(\frac{1}{\epsilon}\right)^{2s}\\
&\leq &
\frac{p^s}{(s-1)!} K \left(\frac{16 e}{\epsilon}\right)^{2s}\\
&\leq&
K p^s  \left(\frac{16 e}{\epsilon}\right)^{2s}.
\end{eqnarray*}
Hence
$$\log N(\epsilon\|F\|_n, \mathcal F, \|\cdot\|_n) \leq s \log p + 2 s \log \left(\frac{16 e}{\epsilon}\right).
$$
  \end{proof}

\begin{proof}[Proof of Theorem \ref{qr.thm}]
We apply Theorem \ref{glm.main2} and take $\psi_\beta(x,y) := \text{sign}(y-x\beta)x.$ Then by Lemma \ref{qr.ee} it follows that $\|\psi_{\hat\beta}\|_\infty = \mathcal O_P(\lambda)$. \\
First note that Condition \ref{itm:ew} is satisfied, which can be seen by direct calculation as follows. 
We have
$$u\mapsto G(u) = \int w(u,y) dP_{\epsilon|x} = \int \text{sign}(y-u) dP_{\epsilon|x} = 1-2F_\epsilon(u).$$
Then $G'(u) = -f_\epsilon(u).$
Hence by the assumed Lipschtiz property of $f_\epsilon,$ it follows that
$$|G'(u)-G'(v)| = | f_\epsilon(u) - f_\epsilon(v)| \leq L |u-v|,$$
thus $G'$ is Lipschitz.\\
Next we need to show that $\mathbb E(w_\beta-w_{\beta_0})^2 = o(n/(s^3(\log p)^2 (\log n)^4)).$
First we calculate the expectation conditioned on $x$:
\begin{eqnarray*}
\mathbb E_{\epsilon|x}(w_\beta-w_{\beta_0})^2 
&=&
\mathbb E_{\epsilon|x} (\text{sign}(y-x^T\beta) - \text{sign}(y-x^T\beta_0))^2
\\
&=&
\int (1_{\epsilon\leq 0 } - 1_{\epsilon \leq x^T(\beta-\beta_0)})^2 dP_{\epsilon|x}
\\
&=&
|\int_0^z dP_{\epsilon|x} | =|F_\epsilon (0)- F_\epsilon (x^T(\beta-\beta_0))| \leq L |x^T(\beta-\beta_0)|.
\end{eqnarray*}
Then for $\beta$ satisfying $\mathbb E\|X(\beta-\beta_0)\|_2^2/n =\mathcal O(s\lambda^2)$ and $\|\beta-\beta_0\|_1=\mathcal O(s\lambda)$ we have
\begin{eqnarray*}
\mathbb E(w_\beta-w_{\beta_0})^2 &=&
 \mathbb E_x(\mathbb E_{\epsilon|x}(w_\beta-w_{\beta_0})^2|x)
\\
& \leq & 
\mathbb E_x L |x^T(\beta-\beta_0)| \leq L \sqrt{\mathbb E_x |x^T(\beta-\beta_0)|^2 }
\\
 &= &
L \sqrt{ (\beta-\beta_0)^T \mathbb E_x  xx^T(\beta-\beta_0) }
\\
&=& \mathcal O(\sqrt{s}\lambda).
\end{eqnarray*}
Then 
$$Pf^2 = P|\Theta_j^Tx|^2 |w_\beta-w_{\beta_0}|^2 \leq s\mathbb E(w_\beta-w_{\beta_0})^2 =\mathcal O({s}^{3/2}\lambda).$$
Then under $s^5 (\log p)^3 (\log n)^4 /n = o(1),$ it follows that $R\sqrt{s\log p}\log n =o(1).$\\
By boundedness of $w_\beta$, we have $$P\sup_f |w_\beta-w_{\beta_0}|^4  = \mathcal O(1) $$ and hence
by the Cauchy-Schwarz inequality, the condition 
$$P\sup_{f}f^4 s^3 (\log p)^3 (\log n)^6=o(1)$$ is satisified.
\\
We have that $\hat{\Theta}'$ is an estimate of $ \Theta'$ constructed using nodewise regression with the matrix $\hat\Sigma:=X^TX/n.$
Then since $\Lambda_{\max}(\mathbb E_x  xx^T) =\mathcal O(1)$, we have by Theorem 2.4 in \cite{vdgeer13}
that
$$\|\hat{\Theta}'_j - \Theta_j'\|_1 = \mathcal O_P(s\lambda).$$
Then since $f_\epsilon(0) \geq c>0$ where $c$ is a universal constant, we get
$$\|\hat\Theta_j - \Theta_j^0\|_1 = \mathcal O_P(s\lambda).$$
Finally, the entropy condition is satisfied by Lemma \ref{qr.entropy}.
\end{proof}

\end{document}